\documentclass[journal,  lettersize]{IEEEtran}
\usepackage{subcaption}
\IEEEoverridecommandlockouts
 \usepackage[english]{babel}
 \usepackage{fancyhdr}
\usepackage{bm}
\usepackage{booktabs}
\usepackage{mathrsfs}
\usepackage{float}
\usepackage{threeparttable}
\usepackage{epsf,epsfig}
\usepackage{amsthm}
\usepackage{amsfonts}
\usepackage[noadjust]{cite}
\usepackage{dsfont}
\usepackage{graphicx,amssymb,amsmath}
\usepackage[dvipsnames]{xcolor}
\usepackage{soul}
\usepackage{algorithmic,algorithm}
\usepackage{url}
\usepackage{xcolor}
\usepackage{subcaption}
\usepackage{color}
\usepackage[colorlinks=true,linkcolor=red,citecolor=red,urlcolor=red,]{hyperref}

\newtheorem{lemma}{\bf Lemma}

\def\BibTeX{{\rm B\kern-.05em{\sc i\kern-.025em b}\kern-.08em
		T\kern-.1667em\lower.7ex\hbox{E}\kern-.125emX}}

\begin{document}

\title{\huge Collaborative Edge AI Inference over Cloud-RAN}

% Kaifeng Han, Guangxu Zhu, Qimei Chen
% D. Wen
% \author{\IEEEauthorblockN{Pengfei Zhang, \IEEEmembership{Student Member, IEEE}, Dingzhu Wen, \IEEEmembership{Senior Member, IEEE}, \\
% Yuanming Shi,
% \IEEEmembership{Senior Member, IEEE}}
\author{\IEEEauthorblockN{Pengfei Zhang,  Dingzhu Wen,  Guangxu Zhu, Qimei Chen, Kaifeng Han, Yuanming Shi}
    \thanks{Pengfei Zhang, Dingzhu Wen and Yuanming Shi are with the Network Intelligence Center, the School of Information Science and Technology, ShanghaiTech University, Shanghai 201210, China (e-mail: \{zhangpf2022, wendzh, shiym\}@shanghaitech.edu.cn), Corresponding author: Dingzhu Wen.}
    \thanks{Guangxu Zhu is with the Shenzhen Research Institute of Big Data, Shenzhen
518172, China (e-mail: gxzhu@sribd.cn).}
    \thanks{Qimei Chen is with the School of Electronic Information, Wuhan University, Wuhan, 430072, China (e-mail: chenqimei@whu.edu.cn).}
    \thanks{Kaifeng Han is with China Academy of Information and
Communications Technology, Beijing 100191, China. (emails: hankaifeng@caict.ac.cn).}
}
\maketitle

% \author{Pengfei Zhang}
% , Dingzhu Wen, Guangxu Zhu, Qimei Chen, Kaifeng Han, Yuanming Shi

% \author{\small \IEEEauthorblockN{Pengfei Zhang\textsuperscript{$\dag$}, Dingzhu Wen\textsuperscript{$\dag$}, Guangxu Zhu\textsuperscript{$\ddag$},
% Qimei Chen\textsuperscript{$\natural$}, Kaifeng Han\textsuperscript{$\sharp$} and Yuanming Shi\textsuperscript{$\dag$} } \\
% 	\textsuperscript{$\dag$} School of Information Science and Technology, ShanghaiTech University, Shanghai, China \\
%         \textsuperscript{$\ddag$} Shenzhen Research Institute of Big Data, Shenzhen, China \\
%         \textsuperscript{$\natural$} School of Electronic Information, Wuhan University, Wuhan, China \\
%         \textsuperscript{$\sharp$} China Academy of Information and Communications Technology, Beijing, China \\
% 	Email: \{zhangpf2022, wendzh, shiym\}@shanghaitech.edu.cn, gxzhu@sribd.cn,  chenqimei@whu.edu.cn, hankaifeng@caict.ac.cn
% 	%\thanks{J. Xu is the corresponding author.}
% }

% \maketitle

\begin{abstract}
In this paper, a cloud radio access network (Cloud-RAN) based collaborative edge AI inference architecture is proposed. Specifically, geographically distributed devices capture real-time noise-corrupted sensory data samples and extract the noisy local feature vectors, which are then aggregated at each remote radio head (RRH) to suppress sensing noise. To realize efficient uplink feature aggregation, we allow each RRH receives local feature vectors from all devices over the same resource blocks simultaneously by leveraging an over-the-air computation (AirComp) technique. Thereafter, these aggregated feature vectors are quantized and transmitted to a central processor (CP) for further aggregation and downstream inference tasks. Our aim in this work is to maximize the inference accuracy via a surrogate accuracy metric called discriminant gain, which measures the discernibility of different classes in the feature space. The key challenges lie on simultaneously suppressing the coupled sensing noise, AirComp distortion caused by hostile wireless channels, and the quantization error resulting from the limited capacity of fronthaul links. To address these challenges, this work proposes a joint  transmit precoding, receive beamforming, and quantization error control scheme to enhance the inference accuracy. Extensive numerical experiments demonstrate the effectiveness and superiority of our proposed optimization algorithm compared to various baselines.

\end{abstract}

\section{Introduction}

\subsection{Overview}

The fundamental purpose of future networks will evolve from delivering conventional human-centric communication services to enable a transformative era of connected intelligence \cite{DBLP:journals/jsac/LetaiefSLL22, DBLP:journals/cm/LetaiefCSZZ19, DBLP:journals/cm/ZhuLDYZH20}. This paradigm shift will empower an array of advanced intelligent services, spanning across diverse domains such as autonomous driving, remote healthcare, and smart city applications, which will be seamlessly accessible at the network edge \cite{DBLP:journals/corr/abs-1911-03878, li2022deep, DBLP:journals/jcin/LanWZZCPH21}.  The implementation of these intelligent services depends on the deployment of well-trained AI models and the utilization of their inference capability for making intelligent decisions, which gives rise to the technique of edge inference \cite{DBLP:journals/comsur/ShiYJZL20,  DBLP:journals/corr/abs-2306-01162, DBLP:journals/tmc/LeeYD23, DBLP:conf/isit/YilmazHG22, zhu2023pushing, DBLP:journals/cm/ShaoZ20}.

Recently, considerable research efforts have been made for the efficient implementation of edge inference \cite{ yang2020energy, huang2020dynamic, yun2021cooperative, he2020attacking, kang2017neurosurgeon, li2019edge}. {\color{blue} Among others, the paradigm of edge-device collaborative inference is the most popular one. Specifically, edge-device collaborative inference divides an AI model into two parts. One part with a small size is deployed at an edge device for feature extraction \cite{DBLP:journals/cm/ShaoZ20} using a method like principal component analysis (PCA). The other computation-intensive part is deployed at the edge server and receives the extracted feature elements from the edge device to complete the residual inference task. It avoids the direct transmission of high-dimensional raw data vectors and offloads most part of the AI model to the server and therefore enjoys the benefits of low communication and computation overhead as well as privacy preservation. } The existing works on edge-device collaborative inference can be divided into two paradigms: single-device and multi-device paradigms. The former paradigm incurs narrow view observations due to a single device’s inherently limited sensing capability \cite{kang2017neurosurgeon,li2019edge,liu2022resource, shi2019improving, shao2021branchy, shao2020bottlenet++, lan2021progressive}. To tackle this issue, the multi-device paradigm has been explored in e.g. \cite{shao2022task,lee2023task,wen2023task,zhuang2023integrated}, where several views of sensory data obtained by multiple devices are collected and fused for inference.

However, the studies on multi-device collaborative inference mainly aim at the cooperation mechanism between multiple devices and the corresponding transceiver design, while ignoring the potential service capability of a single base station (BS).  In fact, the devices at the cell edge may fail to access BS due to weak channel conditions in certain cases \cite{ DBLP:journals/tsp/LiuZ15}. The limited service coverage capability of BSs can be further amplified especially when device mobility is taken into consideration, making it challenging for devices to seamlessly participate in inference tasks. Moreover, the traffic produced by a massive number of devices may also overwhelm a single BS because it exceeds the carrying capacity of BS \cite{DBLP:conf/ewsdn/DawsonMG14}.   To address these limitations and guarantee the inference performance, this paper proposes a  cloud radio access network (Cloud-RAN) \cite{shi2015large} based inference architecture over a resource-constrained wireless network to support the efficient implementation of edge inference.

\subsection{Related Works and Motivations}
One main research focus on edge-device collaborative inference in the single-device context is to further alleviate the computation and communication overhead for enhancing certain performances like achieving ultra-low latency (see e.g., \cite{kang2017neurosurgeon,li2019edge,shi2019improving}). Particularly, a split layer selection strategy is proposed for deep neural networks in \cite{kang2017neurosurgeon} to balance the tradeoff between the communication and computation overhead on devices. Early-exist mechanisms are investigated in \cite{li2019edge,liu2022resource}. The different parts of an AI model are progressively transmitted to the edge device until the accuracy of the current AI sub-model achieves the required performance. Besides, authors in \cite{shi2019improving} develop an efficient and flexible 2-step pruning framework, where unimportant convolution filters in deep neural networks (DNNs) are removed iteratively, and a series of pruned models are generated in the training phase. In addition, other methods, including feature compression techniques (see e.g., \cite{shao2021branchy, shao2020bottlenet++}) and progressive feature transmission \cite{lan2021progressive}, are also proposed.

However, the sensing range of a single device is usually restricted, resulting in a feature that either focuses on a partial view with insufficient information for inference or is extracted from raw data prone to severe distortion. To overcome the limited sensing capability of individual device, multi-device schemes with the target of enhancing the inference performance were proposed in \cite{shao2022task,lee2023task,wen2023task}. In \cite{shao2022task}, a distributed information bottleneck framework was applied to extract and encode features observed by multiple devices from different views of the same target. The local features of the same target that may occur in overlapping areas are captured by multiple devices at \cite{lee2023task}. {\color{blue} A novel multi-view radar sensing scheme was proposed in \cite{wen2023task}, where each device perceives the same wide range of the same target and the server receives aggregated feature vector by over-the-air computation (AirComp) for inference.} Similar to the work in \cite{wen2023task},  \cite{zhuang2023integrated} also assumes homogeneous sensing data and additionally takes the sensing process into consideration.

The above-mentioned works on multi-device paradigm assume that all devices can access the network and be perfectly served by the BS, which is unreasonable especially when handling scenarios where devices face poor channel conditions or mobile traffic surges. As stated in \cite{ma2023over}, simply replicating BS will inevitably result in significant resource waste. Recently, there have been some works in related fields proposed that apply the Cloud-RAN framework to implement federated edge learning (FEEL) to mitigate the above challenges \cite{ma2023over, 10167503}.  \cite{ma2023over} models the global aggregation stage as a lossy distributed source coding problem. \cite{10167503} minimizes the equivalent noise introduced by the FEEL communication stage through the joint design of precoding, quantization, and receive beamforming. Moreover,  \cite{ma2023over} and \cite{10167503} both use the AirComp technique to receive the model update, which greatly improves communication efficiency. Nonetheless, the current implementation of edge inference systems has not taken into account this flexible wireless access network architecture required to support multi-device deployment, which forms the main motivation of our study.

In such an architecture, the BSs are replaced by low-cost and low-power remote radio heads (RRHs), all of which are connected to a centralized processor (CP) located in the baseband unit (BBU) pool through capacity-limited fronthaul links \cite{shi2015large}. The baseband processing is migrated from RRHs to the cloud-computing based CP. RRHs are merely considered as relays with basic signal transmission functionality. As a result, the Cloud-RAN architecture allows the CP to jointly encode or decode user messages thus significantly extending the coverage area \cite{stephen2017joint} and improving inference performance.  However, limited fronthaul capacity between RRHs and  CP  also incurs undesirable quantization error \cite{DBLP:journals/jsac/ZhouY14}. To the best of our knowledge, this work makes the first attempt to apply the Cloud-RAN architecture to complete edge-device collaborative inference. 

\begin{figure}[htbp]
    \centering
    \includegraphics[width=0.45\textwidth]{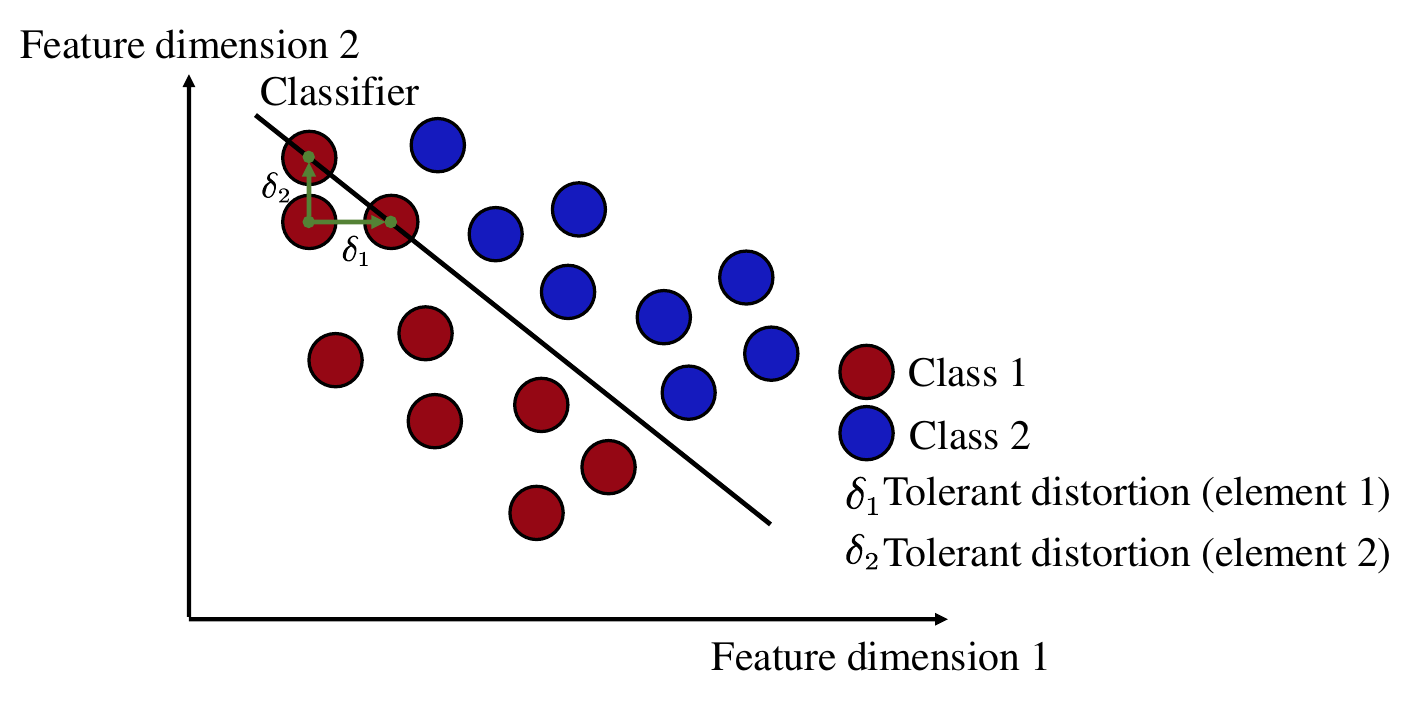}
    \captionsetup{font=small}
    \caption{The varying levels of distortion tolerance among different feature elements in classification tasks.  The distortion level $\delta_1$ can cause incorrect inference on element $1$ but not on element $2$. }\label{fig:dg}
\end{figure} 

On the other hand, as shown by \cite{wen2023task, zhuang2023integrated, shi2023task}, the design of edge inference should feature a task-oriented property. The traditional communication objective to achieve high throughput and low data distortion lacks the ability of distinguishing the feature elements with the same loads and distortion levels but with various importance levels to the inference performance.  In reality, taking the classification task as an example, inference accuracy should be directly maximized as the primary design goal to ensure differential transmission of features. {\color{blue} 
 However, the instantaneous inference accuracy is unknown and lacks of a mathematical model. Recently, there have been some works in the edge inference community that attempt to tackle this problem using an approximate but tractable metric called discriminant gain for classification tasks \cite{lan2021progressive, wen2023task, zhuang2023integrated}.  Discriminant gain is derived based on the well-known Kullback-Leibler (KL) divergence and measures the discernibility of different classes in the feature space.}  For arbitrary two classes, a larger value of discriminant gain represents better separation of the two classes, leading to a higher achievable inference accuracy. For example, a simple classification task is shown in Fig. \ref{fig:dg}, where the feature vector has two feature dimensions. It can be observed that feature dimension 2 is more tolerant to distortion than feature dimension 1 in terms of getting correct inference results. In reality, inference accuracy should be directly maximized as a task goal to ensure the  features 
However, how to apply this metric to the Cloud-RAN based edge inference framework still requires additional study,  which forms the main technical contributions of our paper.

\subsection{Contributions}
In this paper, we propose a  Cloud-RAN based edge inference framework.  The major contributions can be summarized as follows:

\begin{itemize}
        \item \textbf{Cloud-RAN based Multi-device Collaborative Inference System}: We propose a Cloud-RAN architecture based multi-cell network to support 
        multi-device collaborative edge inference system, where a CP serves many geographically distributed devices through multiple RRHs to provide seamless connectivity service. The devices sense a source target from a same wide view to obtain noise-corrupted sensory data for extracting local feature vectors, which are further aggregated by  each RRH using the technique of AirComp. Then, all RRHs quantize their aggregated signals and transmit the compressed signals to the CP, where  all received signals are further aggregated and input into a powerful AI model to finish the downstream inference task.

         \item \textbf{Task-oriented Design Principle}: In the traditional Cloud-RAN based communication system design, most of the works focus on the goal of maximizing the achievable rate, ignoring the task behind the communication. However, in considered edge inference scenario, communication should first serve the inference accuracy, and it is obviously not a wise choice to take achievable rate as the primary goal. To this end, this paper considers a task-oriented design metric, i.e.,  discriminant gain, which can measure the heterogeneous contribution of different feature elements on inference accuracy. By employing this criterion, limited resources can be adaptively allocated to guarantee the most significant feature elements of the inference task can be well received at the CP, leading to an enhanced inference accuracy.

        \item \textbf{Joint Optimization of Quantization, Transmit Precoding, and Receive Beamforming}: Different from existing work where the transmission in different time slots is separately designed, the aggregation of all feature elements is jointly designed. This allows resource allocation among all feature elements, leading to an extra degree of freedom for enhancing the inference accuracy. To this end, a problem of joint quantization noise, transmit precoding and receive beamforming is formulated. To solve this intractable and non-convex problem, we first convert it into an equivalent problem via variable transformation. The equivalent problem then is split into two sub-problems, where one sub-problem is to jointly optimize the receive beamforming and the transmit precoding, and the other sub-problem jointly to optimize the quantization noise matrix and transmit precoding. An iterative algorithm is proposed to solve each sub-problem alternately, where successive convex approximation (SCA) techniques are applied to both sub-problems dealing with the same constraint term.
        
        \item \textbf{Performance Evaluation:} We conduct extensive numerical experiments on a high-fidelity human motion dataset with two inference models, i.e., support vector machine (SVM) and multi-layer perception (MLP) neural network, respectively.    The experiment results demonstrate the effectiveness of the proposed system architecture and optimization approach and also confirm that maximizing discriminant gain indeed improves inference accuracy.
    \end{itemize}

\subsection{Organization and Notations}
 The rest of this paper is organized as follows. Section \uppercase\expandafter{\romannumeral2} describes the system model of Cloud-RAN based multi-devices collaborative inference.  Section \uppercase\expandafter{\romannumeral3} formulates the problems with the goal of maximizing inference accuracy based on the discriminate gain, and simplifies the subsequent analysis by zero-forcing precoding.  An alternating optimization approach is developed in Section \uppercase\expandafter{\romannumeral4} to solve the formulated optimization problem. In Section \uppercase\expandafter{\romannumeral5}, extensive numerical experiments are presented to evaluate the performance of the proposed methods. Finally, Section\uppercase\expandafter{\romannumeral6} concludes this paper. Besides,  Table \ref{tab:abbre} lists some abbreviations used in the paper to facilitate subsequent smooth reading.

% \subsection{Notation}
The notations used in this paper are as follows. The complex and real numbers are denoted by $\mathbb{C}$ and $\mathbb{R}$. The real and imaginary components of complex $x$ are denoted by $\Re$ and $\Im$, respectively. The boldface upper-case letters and boldface lower-case letters represent the matrices and vectors, respectively. The superscripts $(\cdot)^{\sf T}$ and $(\cdot)^{\sf H}$ denotes the transpose and Hermitian operations, respectively. \ $\mathcal{N}(\mathbf{x}; \bf{\mu}, \bf{\Sigma})$ and $\mathcal{CN}(\mathbf{x}; \bf{\mu}, \bf{\Sigma})$ denote that the random variable $\mathbf{x}$ follows Gaussian distribution and complex Gaussian distribution with the mean $\bf{\mu}$ and covariance $\bf{\Sigma}$, respectively.  $\mathbb{E}[\cdot]$ is expectation operator. We use $\mathbf{I}$ and $\text{diag}\left(\{\mathbf{Q}_m\}_{m=1}^M\right) $ respectively denote the identity matrix and diagonal matrix with $\mathbf{Q}_m$ on the diagonal. We let $\mathcal{D}$ to indicate the integer set $\{1, \cdots, D\}$.  For ease of understanding, some important notations and parameters are further summarized in Table \ref{tab:notions}.

\begin{table}[tt]
    \caption{List of abbreviations}
    \centering
    \begin{tabular}{ll}
    \toprule
    Abbreviation & Description \\
    \midrule
    Cloud-RAN & Cloud radio access network \\
    BBU &  baseband unit pool \\
    CP & Central Processor \\
    RRH & Remote radio head \\
    BS & Base station \\
    AirComp & Over-the-air computation \\
    CSI & Channel state information \\
    AWGN & Additive white Gaussian noise \\ 
    FEEL & federated edge learning \\
    DNN & Deep neural network \\
    SVM & support vector machine \\ 
    MLP & multi-layer perception \\
    PCA & Principal component analysis \\
    PDF & probability density function \\
    SCA & successive convex approximation \\
    KL divergence &  Kullback-Leibler  divergence \\
    KKT condition & Karush-Kuhn-Tucker condition \\
    \bottomrule
    \end{tabular}
    \label{tab:abbre}
     \vspace{-1.0em}
\end{table}

\section{System Model}

\begin{table*}[tt]
    \caption{Important Notations}
    \centering
    \begin{tabular}{ll}
    \toprule 
    Notation &  Definition \\
    \midrule
    $K$  &  Number of edge devices \\
    $M$  &  Number of RRHs \\
    $N$  &  Number of antennas for each RRH \\
    $D$  &  Number of feature dimensions (time slots)  \\
    $L$  &  Number of Gaussian components (classes) \\
    $\bm{\mu}_{\ell}$ & Centroid of the ${\ell}$-th class\\
    $\bm{\Sigma}$ & Covariance of all classes  \\
    $C_m$ & Fronthaul link capacity between RRH $m$ and CP \\
    $\mathbf{h}_{k,m}$ & Uplink channel between device $k$ and RRH $m$ \\
    $s_k(d)$ & Uplink transmit signal for device $k$ in the $d$-th time slot \\
    $b_k(d)$ & Uplink precoding scalar for device $k$ in the $d$-th time slot \\ 
    $\mathbf{q}_m(d)$ & Uplink quantization noise for RRH $m$ in the $d$-th time slot \\
    $\mathbf{z}_m(d)$ & Uplink additive white Gaussian noise (AWGN) for RRH $m$ in the $d$-th time slot\\
    $\mathbf{Q}_m$ & Diagonal covariance matrix of $\mathbf{q}_m(d)$ for RRH $m$ \\
    $\mathbf{m}_d$ & Receive beamforming vector in the $d$-th time slot \\
    $\hat{P}$ & Maximum uplink transmit power   \\
    $E$ & Maximum uplink energy consumption \\ 
    \bottomrule
    \end{tabular}
    \label{tab:notions}
\end{table*}

\subsection{Network and Sensing Model}
\begin{figure*}
    \centering
    \includegraphics[scale=0.35]{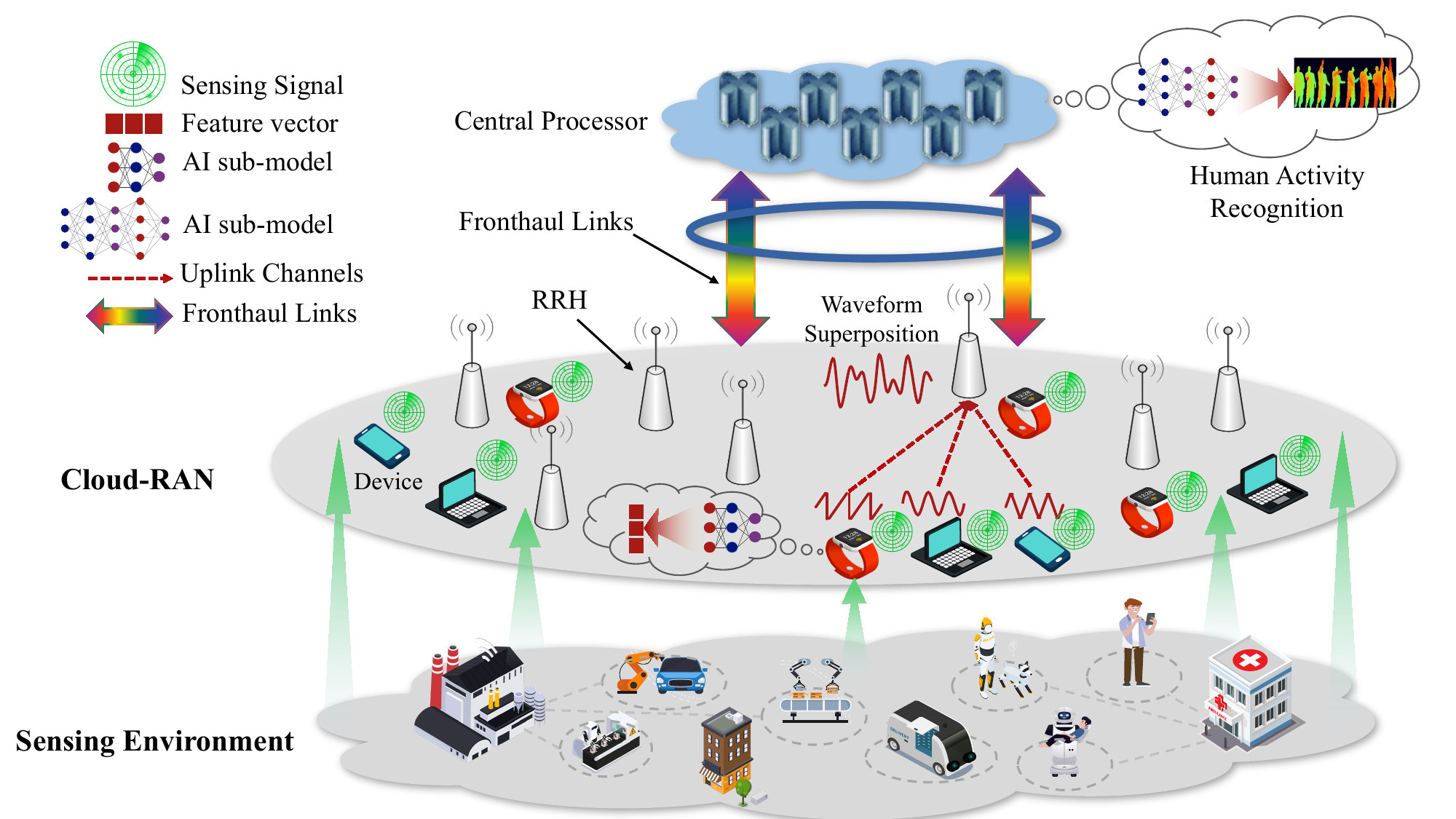}
    \captionsetup{font=small}
    \caption{AirComp-based Cloud-RAN network for edge inference.}
    \label{fig:Cloud-RAN}
\end{figure*}
Consider a multi-cell Cloud-RAN to complete edge inference tasks, where there is one CP, $M$ multi-antenna RRHs, and $K$ single-antenna edge devices. The RRHs lack individual encoding/decoding capability and only have basic signal transmission and reception functions. Each RRH collects information from edge devices via wireless links and then forwards them to CP \cite{DBLP:journals/jsac/ZhouY14, liu2015optimized}. The uplink channel gain between device $k$ and RRH $m$ is denoted as ${\bf h}_{k,m}$.   In uplink transmission, we assume that each device can acquire perfect channel state information (CSI) between itself and all RRHs through uplink pilot signaling \cite{10167503, DBLP:journals/jsac/ZhouY14}. Then the CP serves as a central coordinator, which is also assumed to have the ability to acquire the CSI of all involved links. All RRHs are connected to the CP through a noiseless finite-capacity fronthaul link, as shown in Fig.\ref{fig:Cloud-RAN}. Let $C_m$ denote the fronthaul capacity of the link between RRH $m$ and the CP. The following overall capacity constraint should be satisfied \cite{DBLP:journals/jsac/ZhouY14},
\begin{equation}
    \sum_{m=1}^M C_m \leq C,
\end{equation}
where $C$ is the total capacity of all fronthaul links.

To complete the edge inference task, each device observes the same source target in the same wide view (see e.g., \cite{wen2023task}) to obtain a distortion-corrupted version of the ground-true sensory data.  Then, linear methods like PCA are adopted at each device to extract a local low-dimensional feature vector, which is also noise-corrupted \cite{wen2023task, zhuang2023integrated, wen2023task2, DBLP:journals/tsp/XiaoCLG06, DBLP:journals/tit/XiaoL05}.  Next, each RRH aggregates all feature vectors from all devices to form an intermediate feature vector, which is further quantized and forwarded to the CP via the fronthaul link. At the CP, all intermediate feature vectors are further aggregated to form a global estimate, which is used for finishing the downstream inference task. 

Specifically, the local noise-corrupted sensory data of device $k$ is given by 
\begin{equation} \label{sensory_data}
\begin{aligned}
\mathbf{x}_k = \mathbf{x} + \mathbf{e}_k,
\end{aligned}
\end{equation}
where $\mathbf{x} \in \mathbb{R}^S$ is ground-true sensory data, $\mathbf{e}_k$ is the sensing distortion with the same dimension as the ground-true data.  It is worth noting that wide-view sensing is adopted here, which can be achieved by scanning the sensing directions from angle to angle or conducting beamforming in a MIMO system \cite{yang2017wide}. According to \cite{2006Power},  the sensing distortion vector follows Gaussian distributions with mean zero and covariance $\varepsilon_k^2 \mathbf{I}_k$, i.e., 
\begin{equation}\label{distortion}
\begin{aligned}
\mathbf{e}_k \sim \mathcal{N}(\bm{0}, \varepsilon_k^2 \mathbf{I}),
\end{aligned}
\end{equation}
where $\varepsilon_k^2$ is the sensing noise power.

\subsection{Feature Generation and Distribution}
\begin{figure*}[htbp]
    \centering
    \includegraphics[scale=0.32]{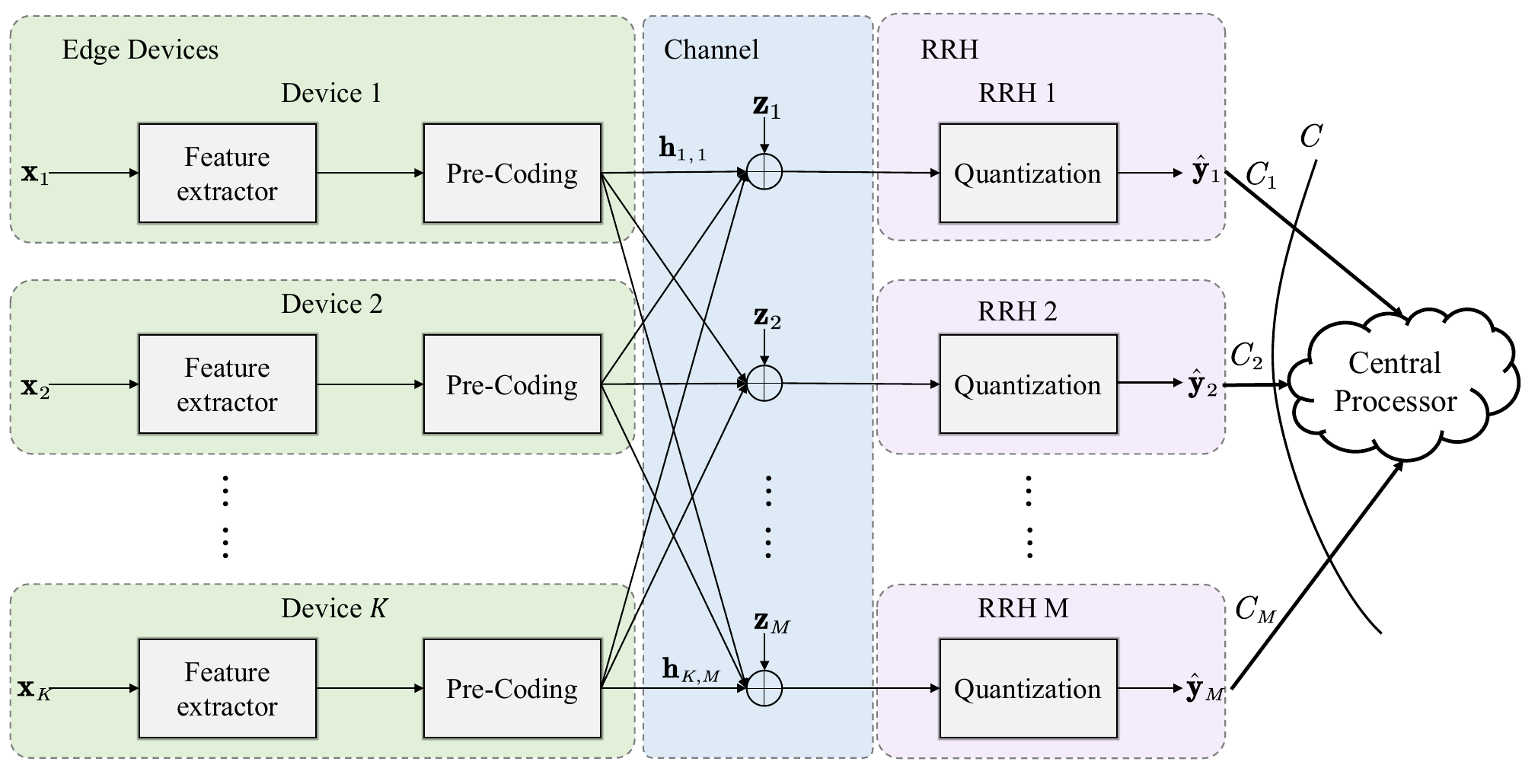}
    \captionsetup{font=small}
    \caption{Illustration of Cloud-RAN system with AirComp.}
    \label{fig:uplink_communication}
\end{figure*}

\subsubsection{Feature Extraction} 
In this work, the method of PCA is used for feature extraction. The detailed procedure is listed below.
\begin{itemize}
    \item In the training stage, the training dataset is used to calculate a principal eigen-space, which is denoted as  $\mathbf{U}$ and satisfies $\mathbf{U}^{\sf T}\mathbf{U} = {\bf I}$, via the eigen-decomposition of the sum co-variance of all data samples. Then, the unitary matrix $\mathbf{U}$ is broadcast to all RRHs and edge devices. 
    
    \item In the inference stage, all local sensory data are projected to the principal eigenspace using ${\bf U}$ for feature extraction.
\end{itemize}
    Specifically, the feature vector extracted at device $K$ can be written as
\begin{equation} \label{feature_vector}
    \mathbf{\Tilde{x}}_k = \mathbf{U}^{\sf T}\mathbf{x}_k =  \mathbf{\Tilde{x}} +  \mathbf{\Tilde{e}}_k  = \mathbf{U}^{\sf T} \mathbf{{x}} + \mathbf{U}^{\sf T}\mathbf{{e}}_k, \ \forall k\in \mathcal{K}. 
\end{equation}
where $ \mathbf{\Tilde{x}} = \mathbf{U}^{\sf T} \mathbf{{x}}$ is the ground-true feature vector, $\mathbf{\Tilde{e}}_k = \mathbf{U}^{\sf T}\mathbf{{e}}_k$ is projected noise vector of edge devices $k$. By leveraging the orthogonality of unitary matrix $\mathbf{U}$, it can be easily shown that the distribution of the projected noise vector remains unchanged, i.e., \begin{equation}\label{Eq:DistortionDistribution}
\mathbf{\Tilde{e}}_k \sim \mathcal{N}(\bm{0}, \varepsilon_k^2 \mathbf{I}).
\end{equation}

\subsubsection{Feature Distribution}
Consider a classification task with $L$ classes. Following the same settings in \cite{lan2021progressive, wen2023task, zhuang2023integrated}, 
we assume that the ground-true feature vector $\mathbf{\Tilde{x}}$ follows a mixture of Gaussian distributions with $L$ Gaussian components. Its probability density function (PDF) is given as
\begin{equation} \label{pdf}
\begin{aligned}
f(\mathbf{\Tilde{x}}) = \frac{1}{L} \sum\limits_{\ell=1}^L \mathcal{N}(\bm{\mu}_{\ell},\bm{\Sigma}),
\end{aligned}
\end{equation} 
where the $\ell$-th Gaussian component $\mathcal{N}(\bm{\mu}_{\ell},\bm{\Sigma})$ corresponds to the $\ell$-th class, $\bm{\mu}_{\ell} \in \mathbb{R}^D$ is the centroid of the $\ell$-th class, $D$ is the dimension of the extracted feature vector, and $\bm{\Sigma} \in \mathbb{R}^{D\times D}$ is a covariance matrix and is same for all classes.  In practice, the raw data or the intermediate feature maps (e.g., the output of a convolutional layer) may not follow a Gaussian mixture model. In this case, a feasible strategy is to fit the data or the feature map into the distribution of the Gaussian mixture. The effectiveness of this approach has been validated through extensive experiments in existing literature \cite{lan2021progressive, wen2023task, zhuang2023integrated, wen2023task2, DBLP:journals/widm/McLachlanR14, mclachlan2019finite}. Since the method of PCA is applied, different elements of the feature vector $\mathbf{\Tilde{x}}$ are independent, i.e., the covariance matrix is diagonal and is denoted as $\bm{\Sigma} = \text{diag}\{\sigma_1^2,\sigma_2^2,...,\sigma_D^2\}$. 

Then, by substituting the distributions of the ground-true feature vector $\mathbf{\Tilde{x}}$ and the sensing distortion $\mathbf{\Tilde{e}}_k$ in  \eqref{pdf} and \eqref{Eq:DistortionDistribution} into the local feature vector $\mathbf{\Tilde{x}}_k$ in \eqref{feature_vector}, we have the following lemma:
\begin{lemma} \label{lma:1}
    The distribution of the local feature vector $\mathbf{\Tilde{x}}_k$ can be derived as
    \begin{equation}
       f( \mathbf{\Tilde{x}}_k ) = \dfrac{1}{L}\sum\limits_{\ell=1}^L \mathcal{N}(\bm{\mu}_{\ell},\bm{\Sigma}+\varepsilon_k^2 \mathbf{I}), \ \forall k\in \mathcal{K}.
    \end{equation}
\end{lemma}
\begin{proof}
Please see Appendix A.
\end{proof}

% \cite{}
\subsection{Communication Model}
To collect all local feature vectors at the CP, the technique of AirComp is adopted to allow all devices transmitting their local feature vectors to the RRHs via a shared multiple access channel, which can significantly enhance the communication efficiency \cite{yang2020federated, zhu2019broadband, DBLP:journals/twc/CaoZXH20}. {\color{blue} In wireless communication, the Aircomp technique is especially suitable for such a scenario where the receiver only focuses on the fusion computation result of massive data from multiple data sources, but does not care about the specific value of each individual data source \cite{liu2020over}. Some examples of computable fusion functions via AirComp can be found in \cite{csahin2023survey}.} As a result, each RRH directly receives an intermediate aggregated analog feature vector, which is further quantized and transmitted to the CP through the assigned fronthaul links, as shown in Fig. \ref{fig:uplink_communication}. The detailed procedure is described as follows.  

\subsubsection{Over-the-Air Aggregation at RRHs}
Since all edge devices are equipped with a single antenna. In each time slot, one dimension of the feature vector is transmitted via AirComp. The whole feature vector with $D$ dimensions is transmitted sequentially over $D$ time slots. Without loss of generality, during the overall $D$ time slots, the channel is assumed to be static, as the time duration of transmitting one symbol is far less than the channel coherence time \cite{zhu2019broadband}. Under this setting, consider an arbitrary time slot $d$, the $d$-th dimension of the feature vector is transmitted by all devices via AirComp. Let $s_k(d) = \mathbf{\Tilde{x}}_k(d)$ denote the transmit signal in the $d$-th time slot, $b_k(d) \in \mathbb{C}$ denote the transmit precoding scalar of edge devices $k$ at the time slot $d$ for the power control. At an arbitrary RRH $m$, the received signal can be derived as
\begin{equation}
\begin{aligned}
\mathbf{y}_m(d) = \sum\limits_{k=1}^K \mathbf{h}_{k,m} b_k(d) s_k(d) + \mathbf{z}_m(d), \; \forall d \in \mathcal{D},
\end{aligned}
\end{equation}
where $\mathbf{h}_{k,m} \in  \mathbb{C}^{N}$ is the channel coefficient between device $k$ and RRH $m$, $N$ denotes the number of antennas on the RRH, and $\mathbf{z}_m(d) \sim \mathcal{CN}(0,\sigma_z^2\mathbf{I})$ denotes the additive white Gaussian noise (AWGN) for RRH $m$. Herein each device's transmit power should not be beyond its maximum transmit power, leading to the following transmit power constraint:
\begin{equation} \label{power_constraint1}
    \mathbb{E} \left[ \left| b_k(d) s_k(d)^2 \right| \right] =  \left| b_k(d) \right|^2 \mathbb{E} \left[ s_k(d)^2  \right] \leq P_k,  \; \forall k\in \mathcal{K}, \; \forall d \in \mathcal{D}.
\end{equation}
Besides, the variance of the transmit signal $s_k(d)$, i.e., $\mathbb{E} \left[ s^2_k(d) \right]$ is known by the CP as a prior
information (e.g., estimated from the offline data samples). Therefore, the power constraint in \eqref{power_constraint1} can be rewritten as
\begin{equation} \label{power_constraint2}
    \left| b_k(d) \right|^2 \leq \hat{P}_k, \; \forall k\in \mathcal{K}, \; \forall d \in \mathcal{D},
    \end{equation}
where $\hat{P}_k = P_k / \mathbb{E} \left[ s^2_k(d) \right]$ is the maximum transmit precoding power. In addition, we also impose a total energy constraint on the data transmission process, that is, the energy consumption of all edge devices in all time slots should satisfy
 \begin{equation} \label{energy_cons1}
 \begin{aligned}
    & \sum\limits_{d=1}^D \sum\limits_{k=1}^K \left( \mathbb{E} \left[ \left| b_k(d) s^2_k(d) \right| \right] \cdot T \right)  \\
    & \quad\quad\quad\quad\quad\quad\quad   =  \sum\limits_{d=1}^D \sum\limits_{k=1}^K \left(\left| b_k(d)  \right|^2  \mathbb{E} \left[ s^2 _k(d) \right] \cdot T \right)  \leq E, 
 \end{aligned}
 \end{equation}
 where $E$ denote the total energy constraint, $T$ is time duration of each AirComp aggregation.
\subsubsection{Quantization of Intermediate Feature Vectors}
The received aggregated intermediate feature vectors $\{{\bf y}_m\}$ are quantized at the RRHs before being forwarded to the CP through the capacity-limited fronthaul links. Each RRH performs the signal quantization independently. The influence of quantization on the signal can be modeled as a Gaussian test channel with the unquantized signals as the input and quantized signals as the output \cite{2015Fronthaul}. Specifically, the $d$-the element of the quantized intermediate feature vector at RRH $m$ can be written as
\begin{equation}
\begin{aligned}
\hat{\mathbf{y}}_m(d) = \mathbf{y}_m(d) + \mathbf{q}_m(d),  \; \forall m \in \mathcal{M}, \; \forall d \in \mathcal{D},
\end{aligned}
\end{equation}
where $\mathbf{q}_m(d) \in \mathbb{C}^N \sim \mathcal{CN}\left(\bm{0}, \mathbf{Q}_m\right)$ denotes the quantization noise and  $\mathbf{Q}_m$ is the diagonal covariance matrix of the quantization noise for RRH $m$ due to independent quantization scheme. Based on the rate-distortion theory \cite{quek2017cloud}, the fronthaul rates of $M$ RRHs at the $d$-th time slot should satisfy
\begin{equation} \label{capacity_constraint1}
\begin{aligned}
&\quad \sum\limits_{m=1}^M C_m(d) = \sum_{m=1}^M I\left({\mathbf{y}}_m(d) ; \hat{\mathbf{y}}_m(d)\right) \\
&= \sum\limits_{m=1}^M \log \frac{\left| \sum_{k=1}^K \left|b_k(d)\right|^2 \mathbf{h}_{k,m} (\mathbf{h}_{k,m})^{\sf H} + \sigma_z^2\mathbf{I} + \mathbf{Q}_m \right|}{\left| \mathbf{Q}_m \right|} \\
&\leq \sum\limits_{m=1}^M \log \frac{\left| \hat{P} \sum_{k=1}^K   \mathbf{h}_{k,m} (\mathbf{h}_{k,m})^{\sf H}  + \sigma_z^2\mathbf{I} + \mathbf{Q}_m
 \right|}{\left| \mathbf{Q}_m \right|} \\
&= \log \frac{\left| \hat{P} \sum_{k=1}^K \mathbf{h}_k  (\mathbf{h}_{k})^{\sf H} + \sigma_z^2\mathbf{I} + \mathbf{Q} \right|}{\left| \mathbf{Q}\right|} \leq C, 
\end{aligned}
\end{equation}
where $\hat{P}$ is maximum transmit power of all edge devices,  $\mathbf{h}_{k} = [\mathbf{h}_{k,1}^{\sf T},\cdots,\mathbf{h}_{k,M}^{\sf T}]^{\sf T}$ is  concatenated channel vector and $\mathbf{Q} = \text{diag}\{\mathbf{Q}_1,\cdots, \mathbf{Q}_m \}$ is defined as the uplink covariance matrix. 
\subsubsection{Global Feature Aggregation at the CP}
The $d$-th element of the received feature vector  at the CP from all RRHs is given by
\begin{equation}
\begin{aligned}
\hat{\mathbf{y}}(d) &= \left[\hat{\mathbf{y}}_1^{\sf T}(d),\cdots,\hat{\mathbf{y}}_M^{\sf T}(d) \right]^{\sf T} \\
& =  \sum\limits_{k=1}^K \mathbf{h}_{k} b_k(d) s_k(d) + \mathbf{z}(d) + \mathbf{q}(d),  \; \forall d \in \mathcal{D},
\end{aligned}
\end{equation}
where $\mathbf{z}(d)=[\mathbf{z}_1^{\sf T}(d),\cdots,\mathbf{z}_M^{\sf T}(d)]^{\sf T}$, $\mathbf{q}(d)=[\mathbf{q}_1^{\sf T}(d),\cdots,\mathbf{q}_M^{\sf T}(d)]^{\sf T}$. To derive a global estimate of the $d$-th element $s(d)$, receive beamforming like in \cite{10167503} is first performed, followed by taking the real part of the processed signal as 
%decode the aggregation signal $\hat{s}(d)$, the CP applies receive beamforming vector to the received $\hat{\mathbf{y}}$'s and take the real part, i.e., 
\begin{equation} \label{aggregation_signal}
\begin{aligned}
\hat{s}(d) &= \mathfrak{R} \left(\mathbf{m}_d^{\sf H} \hat{\mathbf{y}}(d) \right) = \mathfrak{R} \left(\mathbf{m}_d^{\sf H} \sum\limits_{k=1}^K \mathbf{h}_{k} b_k(d) s_k(d)  \right) + \mathbf{n}(d),   
\end{aligned}
\end{equation}
where $\hat{s}(d)$ is the global estimate, $\mathbf{m}_d = [  \mathbf{m}_{d,1}^{\sf T}, \cdots, \mathbf{m}_{d,M}^{\sf T} ]^{\sf T} \in \mathbb{C}^{MN}$ is the receive beamforming vector at time slot $d$, $\mathbf{n}(d) = \mathfrak{R} \left( \mathbf{m}_d^{\sf H}\left( \mathbf{z}(d) + \mathbf{q}(d) \right)\right)$ is the equivalent uplink noise. Given $\mathbf{m}_d$, the equivalent uplink noise is distributed as $\mathbf{n}(d) \sim \mathcal{N}(0,\sigma^2)$ with the variance
\begin{equation}
\begin{aligned}
\sigma^2 = \frac{1}{2} \mathbf{m}_d^{\sf H} \left(\sigma_z^2 \mathbf{I} + \mathbf{Q} \right) \mathbf{m}_d.
\end{aligned}
\end{equation}

\subsection{Discriminant Gain}
As mentioned before, edge inference features task-oriented property as shown in Fig. \ref{fig:dg}, thereby should directly adopt the inference accuracy as the design objective. However, the instantaneous inference accuracy is unknown at the design stage as the input feature is not available on the server. To tackle this problem, an approximate but tractable metric proposed in \cite{lan2021progressive}, called discriminant gain, is adopted as the surrogate for classification tasks. It is derived based on the well-known KL divergence \cite{kullback1997information} and measures the differentiability of different classes in the feature space. Specifically, consider a classification task with $L$ classes, whose ground-true feature distribution is defined in \eqref{pdf}. For an arbitrary pair of classes, say the $\ell$-th and $\ell^{'}$-th classes, the discriminant gain is given by 
\begin{equation}\label{Eq:DG}
\begin{aligned}
G_{\ell,\ell'}(\mathbf{\Tilde{x}}) &= &&{\sf D}_{KL}[\mathcal{N}(\bm{\mu}_{\ell}, \bm{\Sigma}) \ \| \  \mathcal{N}(\bm{\mu}_{\ell'}, \bm{\Sigma})] \\
& &&+ {\sf D}_{KL}[\mathcal{N}(\bm{\mu}_{\ell'}, \bm{\Sigma}) \ \| \  \mathcal{N}(\bm{\mu}_{\ell}, \bm{\Sigma})] \\
&= &&(\bm{\mu}_{\ell} - \bm{\mu}_{\ell'})^{\sf T} \bm{\Sigma}^{-1} (\bm{\mu}_{\ell} - \bm{\mu}_{\ell'}) \\
&=&&\sum\limits_{d=1}^D G_{\ell,\ell'}(\mathbf{\Tilde{x}}(d)), \quad  \forall (\ell, \ell'),
\end{aligned}
\end{equation} 
where $x(d)$ is the $d$-th element of $\mathbf{\Tilde{x}}$ and $G_{\ell,\ell'}(x(d))$ is given as 
\begin{equation}\label{Eq:DGElement}
\begin{aligned}
G_{\ell,\ell'}\left(\mathbf{\Tilde{x}}(d)\right) = \frac{\left( \bm{\mu}_{\ell}(d) -  \bm{\mu}_{\ell'}(d) \right)^2}{\sigma^2_d},  \; \forall d \in \mathcal{D}.
\end{aligned}
\end{equation}
The pair-wise discriminant gain in \eqref{Eq:DG} measures the distance between the class $\ell$ and class $\ell^{'}$ normalized by their covariance. It characterizes the ability of feature vector $\mathbf{\Tilde{x}}$ to distinguish the two classes. In other words, a larger discriminant gain means that the classes are well separated, and thus leading to a higher achievable inference accuracy. Besides, from \eqref{Eq:DGElement}, it is observed different feature elements have different discriminant gains, and thus have heterogeneous contributions on the inference accuracy. To this end, it's desirable to allocate more resources (e.g., power) to make the elements with greater discriminant gains accurately received, which is one of the work's motivations. 

%that the discriminant gains of different feature elements depend on their distributions and are different. 
 
Then, following \cite{lan2021progressive}, the overall discriminant gain is defined as the average of all pair-wise discriminant gains, given as
\begin{equation}
\begin{aligned}
G(\mathbf{\Tilde{x}}) &= \frac{2}{L(L-1)} \sum\limits_{\ell=1}^L \sum\limits_{\ell < \ell'} G_{\ell,\ell'}\left(\mathbf{x}\right) \\
&= \frac{2}{L(L-1)}   \sum\limits_{\ell=1}^L \sum\limits_{\ell < \ell'} \sum\limits_{d=1}^D G_{\ell,\ell'}(\mathbf{x}(d)) \\
&=\sum\limits_{d=1}^D G(\mathbf{\Tilde{x}}(d)), 
\end{aligned}
\end{equation}
where $G(\mathbf{\Tilde{x}}(d))$ is the discriminant gain of the $d$-th feature elements, given as 
\begin{equation}
\begin{aligned}
G(\mathbf{\Tilde{x}}(d)) = \frac{2}{L(L-1)} \sum\limits_{\ell=1}^L \sum\limits_{\ell < \ell'} \frac{\left( \bm{\mu}_{\ell}(d) -  \bm{\mu}_{\ell'}(d) \right)^2}{\sigma^2_d}, \; \forall d \in \mathcal{D}.
\end{aligned}
\end{equation}

\section{Problem Formulation And Simplification}

In this section, we formulated the problem of maximizing discriminant gain under the transmission power and energy constraints of edge devices, as well as the capacity constraint of the fronthaul link between RRH and CP. Subsequently, the transmit side employed a well-known zero-forcing precoding design to derive the distribution of the received features, thereby obtaining a closed-form expression for the discriminant gain in the classification task. This closed-form expression enables the formulated problem to be solved efficiently.

\subsection{Problem Formulation}
For notation simplification, we first define the overall beamforming matrix and scaling matrix  as
\begin{equation}
    \mathbf{M} = \{\mathbf{m}_d, \forall d \in \mathcal{D}\},  \mathbf{B} = \{b_k(d), \forall k \in \mathcal{K}, \forall d \in \mathcal{D}\}.
\end{equation}
Following the task-oriented principle, we aim at maximizing the inference accuracy measured by the overall discriminant gain of the received feature vector at the CP by jointly designing transmit precoding $\mathbf{B}$, on-RRH quantization $\mathbf{Q}$, and on-server beamforming $\mathbf{M}$, as
\begin{equation}
\max_{\mathbf{B},\mathbf{Q}, \mathbf{M}} \;\;
    G =  \sum\limits_{d=1}^D  G\left(\hat{s}_k(d)\right),
\end{equation}
where $\hat{s}_k(d)$  is the $d$-th element of the estimated global feature vector received at the CP defined in \eqref{aggregation_signal}. There are three kinds of constraints, i.e., the transmit power constraints of each device as shown in \eqref{power_constraint2}, the total transmit energy constraint of each device over all times slots as shown in \eqref{energy_cons1}, and the total fronthaul capacity constraint over all RRHs as shown in \eqref{capacity_constraint1}.  Although at first glance the objective function has nothing to do with the optimization variables $\mathbf{B}$, $\mathbf{Q}$, $\mathbf{M}$, the optimization variables influence the objective function by determining the statistical parameters of the estimated global features. In summary, the overall discriminant gain maximization problem is formulated as 

\begin{subequations} \label{obj:original_problem}
\begin{eqnarray}   
\mathscr{P}\!\!:  \mathop{\max}_{\mathbf{B},\mathbf{Q}, \mathbf{M}}   \!\!\!\!\! \!\!\!\!\! 
&& G =  \sum\limits_{d=1}^D  G\left(\hat{s}_k(d)\right)\nonumber \\
\text{ s.t.   }   \!\!\!\!\! \!\!\!\!\!
&& \left| b_k(d) \right|^2 \leq \hat{P}_k, \forall k \in \mathcal{K}, \forall d \in \mathcal{D}, \label{cons:0a}\\
&& \sum\limits_{d=1}^D \sum\limits_{k=1}^K \left(\left| b_k(d)  \right|^2  \mathbb{E} \left[ s^2 _k(d) \right] \cdot T \right) \leq E,  \label{cons:0b}\\
&&  \log \frac{\left| \hat{P} \sum_{k=1}^K \mathbf{h}_k  (\mathbf{h}_{k})^{\sf H} + \sigma_z^2\mathbf{I} + \mathbf{Q} \right|}{\left| \mathbf{Q}\right|}  \leq C \label{cons:0c}. 
\end{eqnarray}
\end{subequations}

 The difficulty in solving the above problem arises from the intractability of the objective function. The derivation of the objective function relies on the distribution of feature elements, which is non-trivial to deal with since it involves the coupling process of precoding, channel, quantization and receive beamforming. To tackle the challenge,  in the following, we will first apply the widely adopted zero-forcing precoding (see, e.g., \cite{wen2019reduced}) design to simplify the problem and thus facilitate the design of subsequent algorithms.

% The difficulty to solve the above problem arises from the following aspects. To begin with, the non-convex and complicated summation form of the objective function makes the optimization variables of different devices and different feature dimensions highly coupled. Besides, the non-convex constraints in \eqref{cons:0a}, \eqref{cons:0b} are non-trivial to deal with, especially the energy constraints involving summation over multiple slots. To tackle these challenges, 
\subsection{Problem Simplification via Zero-Forcing Precoding}

 Without loss of generality, the zero-forcing precoding is adopted to simplify $\mathscr{P}$. Specifically, for each feature dimension $d$, it is given by
\begin{equation}
\mathbf{m}_d^{\sf H} \mathbf{h}_k b_k(d) = c_k(d), \; \forall k\in \mathcal{K}, \; \forall d \in \mathcal{D},
\end{equation}
where we define $\mathbf{C} = \{c_k(d), \forall k \in \mathcal{K}, \forall d \in D\}$ with real-valued element $c_k(d) \geq 0$ representing the receive signal strength from device $k$. Accordingly, the transmit scalar at device $k$ can be derived as
\begin{equation} \label{transmit_scalar_design}
\begin{aligned}
b_k(d) = \frac{c_k(d) (\mathbf{m}_d^{\sf H} \mathbf{h}_k)^{\sf H}}{\left| \mathbf{m}_d^{\sf H} \mathbf{h}_k \right|^2}, \; \forall k\in \mathcal{K}, \; \forall d \in \mathcal{D}.
\end{aligned}
\end{equation}
By substituting the feature vector in \eqref{feature_vector}, the transmission scalar above into $\hat{s}(d)$ in \eqref{aggregation_signal}, it  can be derived as
\begin{equation}\label{estimate}
\begin{aligned}
\hat{s}(d) &= \sum\limits_{k=1}^K c_k(d) s_k(d) + n(d) \\
&=\sum\limits_{k=1}^K c_k(d) \mathbf{\Tilde{x}}(d) + \sum\limits_{k=1}^K c_k(d) \mathbf{\Tilde{e}}_k(d) + n(d).
\end{aligned}
\end{equation}
From \eqref{estimate}, one can observe that the received feature vector is simplified by using the zero-forcing precoding scheme to cancel the interference among different feature elements. This scheme is shown to be effective and is near-optimal when the overall distortion level is low, and is widely adopted in existing designs \cite{wen2019reduced,wiesel2008zero,li2019wirelessly}. Furthermore, as shown in \eqref{estimate}, the zero-forcing precoding scheme allows heterogeneous receive power levels of different feature elements and from different devices. This adaptive power allocation property can provide an extra degree of freedom for enhancing the inference accuracy. 

Based on the simplified form of the received feature vector in \eqref{estimate}, its distribution can be derived as shown in the following lemma.

\begin{lemma}\label{lma:2}
The distribution of the aggregation signal $\hat{s}(d)$ is given by
\begin{equation}
    \begin{aligned}
        \hat{s}(d) \sim \frac{1}{L}\sum\limits_{\ell=1}^L \mathcal{N}\left( \hat{\bm{\mu}}_{\ell}(d), \hat{\sigma}^2_d \right), \; \forall d \in \mathcal{D},
    \end{aligned}
\end{equation}
where the means $\{\hat{\bm{\mu}}_{\ell}(d)\}$ and the variance $\{\hat{\sigma}^2_d\}$ are
\begin{equation} \label{distribution2}
    \left\{  
    \begin{aligned}
        &\hat{\bm{\mu}}_{\ell}(d) = \sum\limits_{k=1}^K c_k(d) \bm{\mu}_{\ell}(d), \\
        &\hat{\sigma}^2_d = \left(\sum\limits_{k=1}^K c_k(d)
         \right)^2 \sigma^2_d  + \sum\limits_{k=1}^K c_k^2(d) \varepsilon_k^2 +  \sigma^2.
    \end{aligned}
    \right.
\end{equation}
\end{lemma}
\begin{proof}
Please see Appendix B.
\end{proof}
It follows that the discriminant gain of the received feature can be derived as
\begin{equation}
    G =  \sum\limits_{d=1}^D  G\left(\hat{s}_k(d)\right) = \frac{2}{L(L-1)} \sum\limits_{d=1}^D  \sum\limits_{\ell=1}^L \sum\limits_{\ell < \ell'} \frac{\left( \hat{\bm{\mu}}_{\ell}(d) -  \hat{\bm{\mu}}_{\ell'}(d) \right)^2}{\hat{\sigma}^2_d}.
\end{equation}
Moreover, the zero-forcing precoding also simplifies the transmit power constraint of each device in the following form. 
\begin{equation} \label{power_constraint3}
\begin{aligned}
\frac{c^2_k(d)}{\hat{P}_k} \leq  \left| \mathbf{m}_d^{\sf H} \mathbf{h}_k \right|^2, \; \forall k\in \mathcal{K}, \; \forall d \in \mathcal{D}.
\end{aligned}
\end{equation}
Likewise,  by substituting the transmission scalar in \eqref{transmit_scalar_design} into the energy constraints of all devices, they can be derived as
\begin{equation} \label{energy_cons2}
\begin{aligned}
    \sum\limits_{d=1}^D \sum\limits_{k=1}^K \frac{c_k^2(d)}{\left| \mathbf{m}_d^{\sf H} \mathbf{h}_k \right|^2} \cdot \mathbb{E} \left[ s^2_k(d)  \right] \leq \frac{E}{T}.
\end{aligned}
\end{equation}
In summary, by applying the zero-forcing precoding, the original discriminant gain maximization problem $\mathscr{P}$  can be simplified as
\begin{subequations} \label{obj:1}
\begin{eqnarray} 
    \mathscr{P}_1\!\!:  \mathop{\text{max}}_{\substack{\mathbf{C},\mathbf{M},  \mathbf{Q}}} \!\!\!\!\! \!\!\!\!\! 
    && G = \frac{2}{L(L-1)} \sum\limits_{d=1}^D  \sum\limits_{\ell=1}^L \sum\limits_{\ell < \ell'} \frac{\left( \hat{\bm{\mu}}_{\ell}(d) -  \hat{\bm{\mu}}_{\ell'}(d) \right)^2}{\hat{\sigma}^2_d} \nonumber\\
    \text{ s.t. } \!\!\!\!\! \!\!\!\!\! 
    && \frac{c_k^2(d)}{\hat{P}_k} \leq \left| \mathbf{m}_d^{\sf H} \mathbf{h}_k \right|^2,  \forall k\in \mathcal{K},  \forall d \in \mathcal{D}, \label{cons:1a}\\
    && \sum\limits_{d=1}^D \sum\limits_{k=1}^K \frac{c_k^2(d)}{\left| \mathbf{m}_d^{\sf H} \mathbf{h}_k \right|^2} \cdot \mathbb{E} \left[ s^2_k(d)  \right] \leq \frac{E}{T},  \label{cons:1b}\\
    &&  \log \frac{\left| \hat{P} \sum_{k=1}^K \mathbf{h}_k  (\mathbf{h}_{k})^{\sf H} + \sigma_z^2\mathbf{I} + \mathbf{Q} \right|}{\left| \mathbf{Q}\right|}  \leq C \label{cons:1c}. 
\end{eqnarray}  
\end{subequations}

The problem in $\mathscr{P}_1$ is still non-convex due to the non-convexity of the objective function and the long-term energy constraints in terms of $c_k(d)$ and ${\bf m}_d$. In the next section, we will illustrate how the simplified problem can be efficiently solved. 
\section{Algorithm Development}

In this section, we developed an efficient algorithm to solve the simplified problem. By applying some variable transformations, the simplified problem is transformed into an equivalent form, which allows us to obtain a sub-optimal solution using successive convex approximation (SCA) and alternating optimization techniques.
Besides, the convergence analysis of the algorithm is also provided at the end.

\subsection{Variables Transformation}
To simplify problem $\mathscr{P}_1$, we introduce auxiliary variables $\mathbf{A} = \{\alpha(d), \cdots, \alpha(d)\}$ with $\alpha(d)$ representing the average discriminant gain  on all class pairs of the $d$-th feature element, which can be given as
\begin{equation} \label{variables:alpha}
\alpha(d) = \frac{2}{L(L-1)}\sum\limits_{\ell=1}^L \sum\limits_{\ell < \ell'} \frac{\left( \hat{\bm{\mu}}_{\ell}(d) -  \hat{\bm{\mu}}_{\ell'}(d) \right)^2}{\hat{\sigma}^2_d},  \; \forall d \in \mathcal{D}.
\end{equation}  
By substituting $\hat{\bm{\mu}}_{\ell}(d)$ and $\hat{\sigma}^2_d$ in \eqref{distribution2} into the constraint \eqref{variables:alpha}, it can be derived as 
\begin{equation} \label{cons:33b}
    \Lambda(\{c_k(d)\}, \{\mathbf{m}_d\}, \mathbf{Q}) = \Gamma_1(\alpha(d), \{c_k(d)\}), 
\end{equation}
where 
\begin{equation}
\begin{aligned}
    &  \Lambda(\{c_k(d)\}, \mathbf{m}_d, \mathbf{Q}) =      \\
    % & \quad \quad  \quad \quad \quad  \quad  \quad \quad \quad \quad + \frac{1}{2} \mathbf{m}_d^{\sf H} \left(\sigma_z^2 \mathbf{I} + \mathbf{Q} \right) \mathbf{m}_d, \; \forall d \in \mathcal{D},  \\
    & \quad \quad \frac{\Big( \sum\limits_{k=1}^K c_k(d) \Big)^2 \sigma^2_d + \sum\limits_{k=1}^K c_k^2(d) \varepsilon_k^2 + \frac{1}{2} \mathbf{m}_d^{\sf H} \left(\sigma_z^2 \mathbf{I} + \mathbf{Q} \right) \mathbf{m}_d}{\frac{2}{L(L-1)} \sum\limits_{\ell=1}^L \sum\limits_{\ell < \ell'}\left({\bm{\mu}}_{\ell}(d)- {\bm{\mu}}_{\ell'}(d) \right)^2}, \\
    & \Gamma_1(\alpha(d), \{c_k(d)\}) = \frac{\Big( \sum\limits_{k=1}^K c_k(d) \Big)^2 }{\alpha(d)}, \; \forall d \in \mathcal{D}.
\end{aligned}
\end{equation}

Next, we can extend the feasible region of the equality constraint \eqref{cons:33b} as below while keeping the same optimal solution to  $\mathscr{P}_1$, which is shown in Lemma $3$.
\begin{equation} \label{cons:333b}
    \Lambda(\{c_k(d)\}, \{\mathbf{m}_d\}, \mathbf{Q}) \leq \Gamma_1(\alpha(d), \{c_k(d)\}). 
\end{equation}
% \begin{multline} \label{cons:333b}
%     \!\!\!\!\!\!\!\!\!\!\!
%     \alpha(d) \leq \frac{\frac{2}{L(L-1)} \Big( \sum\limits_{k=1}^K c_k(d) \Big)^2 \sum\limits_{\ell=1}^L \sum\limits_{\ell < \ell'}\left({\bm{\mu}}_{\ell}(d)- {\bm{\mu}}_{\ell'}(d) \right)^2}{\Big( \sum\limits_{k=1}^K c_k(d) \Big)^2 \sigma^2_d + \sum\limits_{k=1}^K c_k^2(d) \varepsilon_k^2 + \frac{1}{2} \mathbf{m}_d^{\sf H} \left(\sigma_z^2 \mathbf{I} + \mathbf{Q} \right) \mathbf{m}_d},  \\
%     \forall d \in \mathcal{D}.
% \end{multline}
\begin{lemma} \label{lma:3}
   A new problem $ \mathscr{P}_1'$ which extends the feasible region of \eqref{cons:33b} into \eqref{cons:333b} and remaining the same objective function and other constraints reaches the same optimal solution as $ \mathscr{P}_1$.
\end{lemma}
\begin{proof}
    Please see Appendix C.
\end{proof}

Nevertheless, the simplified problem is still very difficult to solve due to the high couple of variables across multiple time slots in the constraints \eqref{cons:1b}. To make this problem feasible, we further introduce auxiliary variables $\mathbf{B} = \left[\beta_{1,1},\beta_{1,2},\cdots, \beta_{k,d} \right]^{\sf T}$ as upper bound such that the following inequality  holds\footnote{The term $\mathbb{E}[s^2_k(d)]$ is omitted just as $\hat{P}_k$ does.}  
\begin{equation} \label{cons:1223}
   \frac{c_k^2(d)}{\left| \mathbf{m}_d^{\sf H} \mathbf{h}_k \right|^2} \leq \beta_{k,d}, \; \; \forall k\in \mathcal{K}, \; \forall d \in \mathcal{D},
\end{equation}
\begin{lemma} \label{lma:4}
    Based on the defined auxiliary variables, the energy constraint term can equivalently be written  as
    \begin{subequations}
    \begin{align}
        &\frac{c_k^2(d)}{\beta_{k,d} } \leq  \left| \mathbf{m}_d^{\sf H} \mathbf{h}_k \right|^2, \; \forall k\in \mathcal{K}, \; \forall d \in \mathcal{D}, \label{cons:2a}\\
        &\sum\limits_{d=1}^D \sum\limits_{k=1}^K \beta_{k,d} \leq E  . \label{cons:2b}
    \end{align}
\end{subequations}
\end{lemma}
\begin{proof}
    Please see Appendix D.
\end{proof}
Therefore, problem \eqref{obj:1} is further reduced to

\begin{subequations} \label{obj:2}
\begin{eqnarray}
\mathscr{P}_2\!\!: 
\mathop{\text{max}}_{\substack{\mathbf{A},\mathbf{B}, \mathbf{C} \\ \mathbf{M},\mathbf{Q}}}  \!\!\!\!\! \!\!\!\!\! 
&& G =  \sum\limits_{d=1}^D  \alpha(d)   \nonumber \\
\text{ s.t. }  \!\!\!\!\! \!\!\!\!\! 
&& 
\frac{c_k^2(d)}{\hat{P}_k}  \leq  \left| \mathbf{m}_d^{\sf H} \mathbf{h}_k \right|^2, \forall k\in \mathcal{K},  \forall d \in \mathcal{D}, \\
&&
\frac{c_k^2(d)}{\beta_{k,d} } \leq \left| \mathbf{m}_d^{\sf H} \mathbf{h}_k \right|^2,  \forall k\in \mathcal{K},  \forall d \in \mathcal{D}, \\
&&
\sum\limits_{d=1}^D \sum\limits_{k=1}^K \beta_{k,d} \leq E, \\
&&  
\log \frac{\left| \hat{P} \sum_{k=1}^K \mathbf{h}_k  (\mathbf{h}_{k})^{\sf H} + \sigma_z^2\mathbf{I} + \mathbf{Q} \right|}{\left| \mathbf{Q}\right|}  \leq C , \\
&&
\Lambda(\{c_k(d)\}, \{\mathbf{m}_d\}, \mathbf{Q}) \leq \Gamma_1(\alpha(d), \{c_k(d)\}),   \nonumber \\
&&
\quad\quad\quad\quad\quad\quad\quad\quad\quad\quad\quad\quad\quad
\forall d \in \mathcal{D}. \label{cons:42d}
  % &&
  % \resizebox{0.8\linewidth}{!}{$\alpha(d) \leq \frac{\frac{2}{L(L-1)} \left( \sum\limits_{k=1}^K c_k(d) \right)^2 \sum\limits_{\ell=1}^L \sum\limits_{\ell < \ell'}\left({\bm{\mu}}_{\ell}(d)- {\bm{\mu}}_{\ell'}(d) \right)^2}{\left( \sum\limits_{k=1}^K c_k(d) \right)^2 \sigma^2_d + \sum\limits_{k=1}^K c_k^2(d) \varepsilon_k^2 + \frac{1}{2} \mathbf{m}_d^{\sf H} \left(\sigma_z^2 \mathbf{I} + \mathbf{Q} \right) \mathbf{m}_d}$}, \nonumber \\
  % &&
  % \quad\quad\quad\quad\quad\quad\quad\quad\quad\quad\quad\quad
  % \forall d \in \mathcal{D}. \label{cons:42d}
\end{eqnarray}
\end{subequations}

\subsection{Alternating Optimization Approach}

In this part, we shall propose an alternating optimization approach to solve problem \eqref{obj:2} for obtaining a suboptimal solution. Specifically, the problem can be split into two subproblems to be solved iteratively. One subproblem fixes the quantization matrix $\mathbf{Q}$ and jointly optimizes the transmit precoding matrix $\mathbf{C}$ and receive beamforming matrix $\mathbf{M}$, while the other fixes other variables and optimizes the quantization matrix $\mathbf{Q}$. The proposed algorithm is summarized in Algorithm \ref{alg:1}.

\subsubsection{Subproblem 1}
With fixed $\mathbf{Q}$, problem \eqref{obj:2} is reduced to the following problem:
\begin{subequations} \label{obj:4}
\begin{eqnarray}
\mathscr{P}_{2.1}\!\!: 
\mathop{\text{max}}_{\substack{\mathbf{A},\mathbf{B} \\ \mathbf{C},   \mathbf{M}}} \!\!\!\!\! \!\!\!\!\! 
&& G =  \sum\limits_{d=1}^D  \alpha(d)  \nonumber \\
\text{ s.t. }  \!\!\!\!\! \!\!\!\!\!
&& 
\frac{c_k^2(d)}{\hat{P}_k} \leq \left| \mathbf{m}_d^{\sf H} \mathbf{h}_k \right|^2, \; \forall k\in \mathcal{K}, \; \forall d \in \mathcal{D}, \\
&&
\frac{c_k^2(d)}{\beta_{k,d} } \leq  \left| \mathbf{m}_d^{\sf H} \mathbf{h}_k \right|^2, \; \forall k\in \mathcal{K}, \; \forall d \in \mathcal{D}, \\
&&
\sum\limits_{d=1}^D \sum\limits_{k=1}^K \beta_{k,d} \leq E, \\
% &&
%   \resizebox{0.8\linewidth}{!}{$\alpha(d) \leq \frac{\frac{2}{L(L-1)} \left( \sum\limits_{k=1}^K c_k(d) \right)^2 \sum\limits_{\ell=1}^L \sum\limits_{\ell < \ell'}\left({\bm{\mu}}_{\ell}(d)- {\bm{\mu}}_{\ell'}(d) \right)^2}{\left( \sum\limits_{k=1}^K c_k(d) \right)^2 \sigma^2_d + \sum\limits_{k=1}^K c_k^2(d) \varepsilon_k^2 + \frac{1}{2} \mathbf{m}_d^{\sf H} \left(\sigma_z^2 \mathbf{I} + \mathbf{Q} \right) \mathbf{m}_d}$}, \nonumber \\
&&
\Lambda(\{c_k(d)\}, \{\mathbf{m}_d\}, \mathbf{Q}) \leq \Gamma_1(\alpha(d), \{c_k(d)\}),   \nonumber \\
&&
\quad\quad\quad\quad\quad\quad\quad\quad\quad\quad\quad\quad\quad
\forall d \in \mathcal{D}. 
\end{eqnarray}
\end{subequations}

\begin{algorithm}
 \caption{Proposed Algorithm for Solving Problem $\mathscr{P}$}
 \begin{algorithmic}[1]  \label{alg:1}
    \REQUIRE {Initial points $\mathbf{A}^{[0]}$, $\mathbf{B}^{[0]}$,$\mathbf{C}^{[0]}$, $\mathbf{M}^{[0]}$, $\mathbf{Q}^{[0]}$ and solution precision $\epsilon$.}
    \STATE {Set $t = 0$.}
    \REPEAT
    \STATE {Solving problem 
    \eqref{obj:6} for given $\mathbf{Q}^{[t]}$, and denote the updating solution as $\{\mathbf{A}^{[t+1/2]}, \mathbf{B}^{[t+1]},
    \mathbf{C}^{[t+1]}, \mathbf{M}^{[t+1]}\}$;}
    \STATE {Solving problem 
    \eqref{obj:7} for given $\mathbf{C}^{[t+1]}$, $\mathbf{M}^{[t+1]}$, and denote the updating solution as $\{\mathbf{A}^{[t+1]}, \ \mathbf{Q}^{[t+1]}\}$;}
    \STATE {Compute discirminant gain $G$;}
    \STATE {Set $t = t+1$;}
     \UNTIL{the increase of the discriminant gain is below the given threshold $\epsilon$.}
    \ENSURE $\mathbf{A}$, $\mathbf{B}$, $\mathbf{C}$, $\mathbf{M}$ and $\mathbf{Q}$.
 \end{algorithmic}
\end{algorithm}
 
Although the objective function is convex, it is still challenging to solve problem \eqref{obj:4} due to the non-convex constraints. In general, there is no standard method for solving such non-convex optimization problems optimally. Herein we adopt the SCA technique to solve problem \eqref{obj:4}. To apply the SCA approach, we convert problem \eqref{obj:4} from the complex domain to the real domain with the following variables:
\begin{subequations}
    \begin{align}
        &\tilde{\mathbf{m}}_d =  \left[\Re(\mathbf{m}_d)^{\sf T}, \Im(\mathbf{m}_d)^{\sf T}\right]^{\sf T},  \forall d \in \mathcal{D}, \\
       &\tilde{\mathbf{H}}_k = 
       \begin{bmatrix}
            \Re(\mathbf{h}_k \mathbf{h}_k^{\sf H}) &     -\Im(\mathbf{h}_k \mathbf{h}_k^{\sf H}) \\
            \Im(\mathbf{h}_k \mathbf{h}_k^{\sf H}) &
            \Re(\mathbf{h}_k \mathbf{h}_k^{\sf H})
       \end{bmatrix}, 
       \forall k \in \mathcal{K}, \\
        &\tilde{\mathbf{Q}} = 
        \begin{bmatrix}
            \Re(\tilde{\mathbf{Q}}) & -\Im(\tilde{\mathbf{Q}}) \\
            \Im(\tilde{\mathbf{Q}}) &
            \Re(\tilde{\mathbf{Q}}) \\
        \end{bmatrix}.
    \end{align}
\end{subequations}
The problem \eqref{obj:4} can be reformulated as follows:
\begin{subequations} \label{obj:55}
\begin{eqnarray}
\mathop{\text{max}}_{\substack{\mathbf{A},\mathbf{B} \\ \mathbf{C}, \tilde{\mathbf{M}}}}  \!\!\!\!\! \!\!\!\!\! 
&& G =  \sum\limits_{d=1}^D  \alpha(d)  \nonumber\\
\text{ s.t. } \!\!\!\!\! \!\!\!\!\! 
% && \text{constraints } \eqref{cons:2b}, \nonumber\\
&& \frac{c_k^2(d)}{\hat{P}_k} \leq  \tilde{\mathbf{m}}_d^{\sf T}\tilde{\mathbf{H}}_k \tilde{\mathbf{m}}_d, \forall k \in \mathcal{K}, \forall d \in \mathcal{D},\\
&& \frac{c_k^2(d)}{\beta_{k,d} } \leq  \tilde{\mathbf{m}}_d^{\sf T}\tilde{\mathbf{H}}_k \tilde{\mathbf{m}}_d, \forall k \in \mathcal{K}, \forall d \in \mathcal{D}, \\
&& \sum\limits_{d=1}^D \sum\limits_{k=1}^K \beta_{k,d} \leq E, \\
% &&  \resizebox{0.92\linewidth}{!}{$\frac{\left( \sum\limits_{k=1}^K c_k(d) \right)^2 \sigma_d^2 + \sum\limits_{k=1}^K c_k^2(d) \varepsilon_k^2 + \frac{1}{2} \tilde{\mathbf{m}}_d^{\sf H} \left(\sigma_z^2 \mathbf{I} + \tilde{\mathbf{Q}} \right) \tilde{\mathbf{m}}_d}{\frac{2}{L(L-1)}\sum\limits_{\ell=1}^L \sum\limits_{\ell < \ell'}\left({\bm{\mu}}_{\ell}(d)- {\bm{\mu}}_{\ell'}(d) \right)^2}  
% $} \nonumber \\ 
% &&
%   \quad\quad\quad\quad\quad\quad\quad\quad\quad \leq \frac{\left( \sum\limits_{k=1}^K c_k(d) \right)^2}{\alpha(d)}, \forall d \in \mathcal{D}.
&&
\Lambda(\{c_k(d)\}, \{\tilde{\mathbf{m}}_d\}, \mathbf{Q}) \leq \Gamma_1(\alpha(d), \{c_k(d)\}),   \nonumber \\
&&
\quad\quad\quad\quad\quad\quad\quad\quad\quad\quad\quad\quad\quad\quad\quad
\forall d \in \mathcal{D}. 
\end{eqnarray}
\end{subequations}
Next we also define  
\begin{equation}
    \Gamma_2(\tilde{\mathbf{m}}_d) = \tilde{\mathbf{m}}_d^{\sf T}\tilde{\mathbf{H}}_k \tilde{\mathbf{m}}_d, \forall k \in \mathcal{K}, \forall d \in \mathcal{D}.
\end{equation}
% \begin{equation}
%         f_{k,d}(\tilde{\mathbf{m}}_d) =  \tilde{\mathbf{m}}_d^{\sf T}\tilde{\mathbf{H}}_k \tilde{\mathbf{m}}_d, \forall k \in \mathcal{K}, \forall d \in \mathcal{D},\\
% \end{equation}
% \begin{equation}
%         g_{d}\left(\{c_k\},\alpha(d)\right) = \frac{\left( \sum\limits_{k=1}^K c_k(d) \right)^2}{\alpha(d)} ,  \forall d \in \mathcal{D}.
% \end{equation}
and then the following lemma is obtained.
\begin{lemma}
% Given the reference point  ${\mathbf{A}}^{[t]}, {\mathbf{C}}^{[t]}, \tilde{\mathbf{M}}^{[t]}$ in  the $t$-th iteration, the function $\{f_{k,d}(\tilde{\mathbf{m}}_d)\}$, $\{ g_{d}\left(\{c_k\},\alpha(d)\right)\}$ is lower bounded by their respective first-order Taylor expansion, i.e., 
Given the reference point  ${\mathbf{A}}^{[t]}, {\mathbf{C}}^{[t]}, \tilde{\mathbf{M}}^{[t]}$ in  the $t$-th iteration, the function $\Gamma_1(\alpha(d), \{c_k(d)\})$, $\Gamma_2(\tilde{\mathbf{m}}_d)$ is lower bounded by their respective first-order Taylor expansion, i.e., 
% \begin{equation}
% \begin{aligned}
%    & f_{k,d}(\tilde{\mathbf{m}}_d) \geq \hat{f}_{k,d}(\tilde{\mathbf{m}}_d^{[t]})
%    \\ = & {f}_{k,d}(\tilde{\mathbf{m}}_d^{[t]}) + \nabla_{\tilde{\mathbf{m}}_d} {f}_{k,d}(\tilde{\mathbf{m}}_d^{[t]})^{\sf T}(\tilde{\mathbf{m}}_d- \tilde{\mathbf{m}}_d^{[t]}) \\
%    = &
%    (2\tilde{\mathbf{H}}_k\tilde{\mathbf{m}}_d^{[t]})^{\sf T} \tilde{\mathbf{m}}_d - (\tilde{\mathbf{m}}_d^{[t]})^{\sf T}\tilde{\mathbf{H}}_k\tilde{\mathbf{m}}_d^{[t]}, \forall k \in \mathcal{K}, \forall d \in \mathcal{D},\\
% \end{aligned}
% \end{equation}
% \begin{equation}
% \begin{aligned}
%       \!\!
%       & {g}_{d}(\{c_k\},\alpha(d)) \geq \hat{g}_{d} (\{c_k^{[t]}\},\alpha^{[t]}(d))\\
%       =& {g}_{d}(\{c_k^{[t]}\},\alpha^{[t]}(d)) +  \nabla_{\alpha(d)}{g}_{d}(\{c_k^{[t]}\},\alpha^{[t]}(d))(\alpha(d) - \alpha^{[t]}(d))  \\
%       & +  \sum\limits_{k=1}^K  \nabla_{c_k(d)}{g}_{d}(\{c_k^{[t]}\},\alpha^{[t]}(d))(c_k(d) - c_k^{[t]}(d)) , \forall d \in \mathcal{D},
% \end{aligned}
% \end{equation}
% \text{where} 
% \begin{equation}
% \begin{aligned}
%     & \nabla_{\alpha(d)}{g}_{d}(\{c_k^{[t]}\},\alpha^{[t]}(d)) = - \Bigg( \frac{ \sum\limits_{k=1}^K c_k^{[t]}(d)}{\alpha^{[t]}(d)} \Bigg)^2, \forall d \in \mathcal{D}, \\
%     & \nabla_{c_k(d)}{g}_{d}(\{c_k^{[t]}\},\alpha^{[t]}(d)) =  \frac{2\sum\limits_{k=1}^K c_k^{[t]}(d) }{\alpha^{[t]}(d)}, \forall d \in \mathcal{D}.
% \end{aligned}
% \end{equation}
\begin{equation}
\begin{aligned}
    \quad
    & \Gamma_1(\alpha(d), \{c_k(d)\})  \geq \hat{\Gamma}_1(\alpha^{[t]}(d), \{c_k^{[t]}\}) \\
    = \;\; & \Gamma_1(\alpha^{[t]}(d), \{c_k^{[t]}\}) + \frac{\partial \Gamma_1(\alpha^{[t]}(d), \{c_k^{[t]}\})}{\partial \alpha(d)} (\alpha(d) - \alpha^{[t]}(d)) \\ 
    \;\; & + \sum\limits_{k=1}^K  \frac{\partial \Gamma_1(\alpha^{[t]}(d), \{c_k^{[t]}\})}{\partial c_k(d)} (c_k(d) - c_k^{[t]}(d)), \forall d \in \mathcal{D},
\end{aligned}
\end{equation}
\text{where} 
\begin{equation}
\begin{aligned}
    & \frac{\partial \Gamma_1(\alpha^{[t]}(d), \{c_k^{[t]}\})}{\partial \alpha(d)} = - \Bigg( \frac{ \sum\limits_{k=1}^K c_k^{[t]}(d)}{\alpha^{[t]}(d)} \Bigg)^2, \forall d \in \mathcal{D}, \\
    & \frac{\partial \Gamma_1(\alpha^{[t]}(d), \{c_k^{[t]}\})}{\partial c_k(d)} =  \frac{2\sum\limits_{k=1}^K c_k^{[t]}(d) }{\alpha^{[t]}(d)}, \forall d \in \mathcal{D}.
\end{aligned}
\end{equation}
\begin{equation}
\begin{aligned}
    \quad
   & \Gamma_2(\tilde{\mathbf{m}}_d) \geq  \hat{\Gamma}_2(\tilde{\mathbf{m}}_d^{[t]}) \\
   = \;\; & \Gamma_2(\tilde{\mathbf{m}}_d^{[t]})  + \frac{\partial \Gamma_2(\tilde{\mathbf{m}}_d)}{\partial  \tilde{\mathbf{m}}_d }  (\tilde{\mathbf{m}}_d- \tilde{\mathbf{m}}_d^{[t]}) \\ 
   = \;\; & 
   (2\tilde{\mathbf{H}}_k\tilde{\mathbf{m}}_d^{[t]})^{\sf T} \tilde{\mathbf{m}}_d - (\tilde{\mathbf{m}}_d^{[t]})^{\sf T}\tilde{\mathbf{H}}_k\tilde{\mathbf{m}}_d^{[t]}, \forall k \in \mathcal{K}, \forall d \in \mathcal{D}.
\end{aligned}
\end{equation}
\end{lemma}
With any given local point $\{\mathbf{A}^{[t]}, \mathbf{C}^{[t]},  \tilde{\mathbf{M}}^{[t]}\}$ as well as the lower
bounds, problem  \eqref{obj:55} is approximated as the following problem in \eqref{obj:6}, whose feasible region is a subset of the problem in \eqref{obj:55}.
\begin{subequations} \label{obj:6}
\begin{eqnarray}
% &\mathop{\text{maximize}}_{\substack{\{\alpha(d)\},\{\beta_{k,d}\}, \\ \{c_k(d)\}, \{\tilde{\mathbf{m}}_d\}}} 
\mathop{\text{max}}_{\substack{\mathbf{A},\mathbf{B} \\ \mathbf{C}, \tilde{\mathbf{M}}}}  \!\!\!\!\! \!\!\!\!\! 
&& G =  \sum\limits_{d=1}^D  \alpha(d)  \nonumber\\
\text{ s.t. }  \!\!\!\!\! \!\!\!\!\!
% && \frac{c_k^2(d)}{\hat{P}_k} \leq  \hat{f}_{k,d}(\tilde{\mathbf{m}}_d^{[t]})
&& \frac{c_k^2(d)}{\hat{P}_k} \leq \hat{\Gamma}_2(\tilde{\mathbf{m}}_d^{[t]}), \forall k \in \mathcal{K}, \forall d \in \mathcal{D},\\
% && \frac{c_k^2(d)}{\beta_{k,d}} \leq  \hat{f}_{k,d}(\tilde{\mathbf{m}}_d^{[t]}), \forall k \in \mathcal{K}, \forall d \in \mathcal{D},  \\
&& \frac{c_k^2(d)}{\beta_{k,d}} \leq  \hat{\Gamma}_2(\tilde{\mathbf{m}}_d^{[t]}), \forall k \in \mathcal{K}, \forall d \in \mathcal{D},  \\
&& \sum\limits_{d=1}^D \sum\limits_{k=1}^K \beta_{k,d} \leq E,  \\
% &&  \resizebox{0.92\linewidth}{!}{$\frac{\left( \sum\limits_{k=1}^K c_k(d) \right)^2 \sigma_d^2 + \sum\limits_{k=1}^K c_k^2(d) \varepsilon_k^2 + \frac{1}{2} \tilde{\mathbf{m}}_d^{\sf H} \left(\sigma_z^2 \mathbf{I} + \tilde{\mathbf{Q}} \right) \tilde{\mathbf{m}}_d}{  \frac{2}{L(L-1)} \sum\limits_{\ell=1}^L \sum\limits_{\ell < \ell'} \left({\bm{\mu}}_{\ell}(d)- {\bm{\mu}}_{\ell'}(d) \right)^2}$}   \nonumber  \\
% && \quad\quad\quad\quad\quad\quad\quad  \leq \hat{g}_{d}\left(\{c_k^{[t]}\},\alpha^{[t]}(d)\right), \forall d \in \mathcal{D}.
&&
\Lambda(\{c_k(d)\}, \{\mathbf{m}_d\}, \mathbf{Q}) \leq \hat{\Gamma}_1(\alpha^{[t]}(d), \{c_k^{[t]}\}),   \nonumber \\
&&
\quad\quad\quad\quad\quad\quad\quad\quad\quad\quad\quad\quad\quad
\forall d \in \mathcal{D}. 
\end{eqnarray}
\end{subequations}
As a result, this problem is convex, which can be efficiently solved by using convex optimization tools, e.g., CVX \cite{cvx}.
\subsubsection{Subproblem 2}
Next we fix transmit precoding matrix $\mathbf{C}$ and the receive beamforming matrix $\mathbf{M}$ to optimize the quantization noise matrix, then problem \eqref{obj:2} is reduced to the following problem:
% When fixing the receive beamforming vector $\{\mathbf{m}_d\}$, problem \eqref{obj:2} is reduced to the following problem:
\begin{subequations} \label{obj:7} 
\begin{eqnarray}
\mathscr{P}_{2,2}\!\!: 
\mathop{\text{max}}_{\mathbf{A},  \mathbf{Q}}  \!\!\!\!\! \!\!\!\!\!
&& G =  \sum\limits_{d=1}^D  \alpha(d)  \nonumber \\
\text{ s.t. }  \!\!\!\!\! \!\!\!\!\!
&&  
\log \frac{\left| \hat{P} \sum_{k=1}^K \mathbf{h}_k  (\mathbf{h}_{k})^{\sf H} + \sigma_z^2\mathbf{I} + \mathbf{Q} \right|}{\left| \mathbf{Q}\right|}  \leq C, \\
% &&
% \resizebox{0.92\linewidth}{!}{$\frac{\left( \sum\limits_{k=1}^K c_k(d) \right)^2 \sigma_d^2 + \sum\limits_{k=1}^K c_k^2(d) \varepsilon_k^2 + \frac{1}{2} {\mathbf{m}}_d^{\sf H} \left(\sigma_z^2 \mathbf{I} + {\mathbf{Q}} \right) {\mathbf{m}}_d}{\frac{2}{L(L-1)}\sum\limits_{\ell=1}^L \sum\limits_{\ell < \ell'}\left({\bm{\mu}}_{\ell}(d)- {\bm{\mu}}_{\ell'}(d) \right)^2}$} \nonumber \\
% && \quad\quad\quad\quad\quad\quad\quad\quad\quad \leq  \frac{\left( \sum\limits_{k=1}^K c_k(d) \right)^2}{\alpha(d)}, \forall d \in \mathcal{D}. 
&&
\Lambda(\{c_k(d)\}, \{\mathbf{m}_d\}, \mathbf{Q}) \leq \Gamma_1(\alpha(d), \{c_k(d)\}),   \nonumber \\
&&
\quad\quad\quad\quad\quad\quad\quad\quad\quad\quad\quad\quad\quad
\forall d \in \mathcal{D}. 
\end{eqnarray}
\end{subequations}
It is not hard to verify that all constraints in \eqref{obj:7} are convex with respect to $\mathbf{Q}$ \cite{DBLP:journals/jsac/ZhouY14}. For auxiliary variables $\mathbf{A}$, we apply similar SCA technique to $\Gamma_1(\alpha(d), \{c_k(d)\})$ but Taylor expansion is only done at $\mathbf{A}$. Therefore, this problem also becomes convex.

\subsection{Complexity and Convergence Analysis}
 As the complexity of alternating optimization is difficult to achieve, the complexity of solving the subproblem in each iteration is analyzed. The complexity of the (47) subproblem is bounded by $\mathcal{O}((2K+MN+1)^3D^3)$, where  $(2K+MN+1)D$ is the number of variables. The complexity of the (48) subproblem is given by  $\mathcal{O}((MN+D)^3)$, where  $MN+D$ is the number of variables.

 Based on \cite{marks1978general}, it can be proved that the solutions of problems \eqref{obj:6} and \eqref{obj:7} will eventually converge to a stationary point that satisfies the Karush-Kuhn-Tucker (KKT) conditions. Similar conclusions are also derived in those works based on SCA and alternating optimization \cite{DBLP:journals/twc/SunNW21, DBLP:journals/tvt/LyuXZZ23}. The complete proof process is omitted here due to space limitation. Next, we focus on the convergence of alternating optimization.
We denote $G(\mathbf{A}, \mathbf{B}, \mathbf{C}, \mathbf{M}, \mathbf{Q})$ as the value of the objective function in problem \eqref{obj:2} for a feasible solution $\{\mathbf{A}, \mathbf{B}, \mathbf{C}, \mathbf{M}, \mathbf{Q}\}$. As shown in step $4$ of Algorithm \ref{alg:1}, a feasible solution of problem \eqref{obj:7} (i.e., $\{\mathbf{A}^{[t]}, \mathbf{B}^{[t]}, \mathbf{C}^{[t]}, \mathbf{M}^{[t]}, \mathbf{Q}^{[t]}\}$)  is also  feasible to problem \eqref{obj:6}. 
The reasons are as follows. In problem \eqref{obj:7}, only the auxiliary variable $\mathbf{A}$ and quantization noise matrix $\mathbf{Q}$ are optimized with constraint \eqref{cons:42d} still being satisfied. Besides, for the optimized precoding $\mathbf{C}$ and beamforming matrix $\mathbf{M}$ of problem \eqref{obj:6}, the remaining constraints as well as the constraint in \eqref{cons:42d} also hold, such that a feasible solution of problem \eqref{obj:7} is always feasible for problem \eqref{obj:6}.   We denote $\{\mathbf{A}^{[t]}, \mathbf{B}^{[t]}, \mathbf{C}^{[t]}, \mathbf{M}^{[t]}, \mathbf{Q}^{[t]}\}$ and $\{\mathbf{A}^{[t+1]}, \mathbf{B}^{[t+1]}, \mathbf{C}^{[t+1]}, \mathbf{M}^{[t+1]}, \mathbf{Q}^{[t+1]}\}$ as a feasible solution of problem \eqref{obj:2} at the $t$-th and $(t+1)$-th iterations, respectively.

Then, for step $3$ of Algorithm \ref{alg:1}, problem \eqref{obj:6} is convex under given $\mathbf{Q}^{[t]}$, solving which leads to a non-decreasing value of the objective function, i.e., 
% \begin{equation} \label{ineq:1}
% \begin{aligned}
%     G(\mathbf{A}^{[t]}, 
%     \mathbf{B}^{[t]},
%     \mathbf{C}^{[t]}, \mathbf{M}^{[t]}, \mathbf{Q}^{[t]}) \leq      G(\mathbf{A}^{[t+1/2]}, 
%     \mathbf{B}^{[t+1]},
%     \mathbf{C}^{[t+1]}, \mathbf{M}^{[t+1]}, \mathbf{Q}^{[t]}),
% \end{aligned} 
% \end{equation}
\begin{multline} \label{ineq:1}
        G(\mathbf{A}^{[t]}, 
    \mathbf{B}^{[t]},
    \mathbf{C}^{[t]}, \mathbf{M}^{[t]}, \mathbf{Q}^{[t]}) \leq     \\
    G(\mathbf{A}^{[t+1/2]}, 
    \mathbf{B}^{[t+1]},
    \mathbf{C}^{[t+1]}, \mathbf{M}^{[t+1]}, \mathbf{Q}^{[t]}),
\end{multline}
where $\{\mathbf{A}^{[t+1/2]}, 
 \mathbf{B}^{[t+1]}, \mathbf{C}^{[t+1]}, \mathbf{M}^{[t+1]}\}$ is the solution obtained by solving problem \eqref{obj:6} using convex approximation technique. Similarly, for given $\mathbf{C}^{[t+1]}$ and $\mathbf{M}^{[t+1]}$ as shown in step $4$ of Algorithm \ref{alg:1},  the solution  $\{\mathbf{A}^{[t+1]},  \mathbf{Q}^{[t+1]}\}$ obtained
  %  by one-step iteration 
  on problem \eqref{obj:7} will also not reduce the value of the objective function, thus we have
% \begin{equation} \label{ineq:2}
% \begin{aligned}
%     G(\mathbf{A}^{[t+1/2]}, 
%     \mathbf{B}^{[t+1]},
%     \mathbf{C}^{[t+1]}, \mathbf{M}^{[t+1]}, \mathbf{Q}^{[t]}) \leq      G(\mathbf{A}^{[t+1]}, 
%     \mathbf{B}^{[t+1]},
%     \mathbf{C}^{[t+1]}, \mathbf{M}^{[t+1]}, \mathbf{Q}^{[t+1]}).
% \end{aligned} 
% \end{equation}
\begin{multline} \label{ineq:2}
    G(\mathbf{A}^{[t+1/2]}, 
    \mathbf{B}^{[t+1]},
    \mathbf{C}^{[t+1]}, \mathbf{M}^{[t+1]}, \mathbf{Q}^{[t]}) \leq  \\
    G(\mathbf{A}^{[t+1]}, 
    \mathbf{B}^{[t+1]},
    \mathbf{C}^{[t+1]}, \mathbf{M}^{[t+1]}, \mathbf{Q}^{[t+1]}).
\end{multline}
Based on \eqref{ineq:1} and \eqref{ineq:2}, we further obtain
% \begin{equation}
% \begin{aligned}
%     G(\mathbf{A}^{[t+1]}, 
%     \mathbf{B}^{[t+1]},
%     \mathbf{C}^{[t+1]}, \mathbf{M}^{[t+1]}, \mathbf{Q}^{[t+1]}) 
%     \geq
%     G(\mathbf{A}^{[t]}, 
%     \mathbf{B}^{[t]},
%     \mathbf{C}^{[t]}, \mathbf{M}^{[t]}, \mathbf{Q}^{[t]}),
% \end{aligned}
% \end{equation}
\begin{multline}
    G(\mathbf{A}^{[t+1]}, 
    \mathbf{B}^{[t+1]},
    \mathbf{C}^{[t+1]}, \mathbf{M}^{[t+1]}, \mathbf{Q}^{[t+1]}) 
    \geq \\
    G(\mathbf{A}^{[t]}, 
    \mathbf{B}^{[t]},
    \mathbf{C}^{[t]}, \mathbf{M}^{[t]}, \mathbf{Q}^{[t]}),
\end{multline}
which shows that the objective value of problem \eqref{obj:2} is always increasing over iterations.
Therefore, the proposed algorithm converges. This thus completes the proof.

\section{Numerical Results}
In this section, we evaluate the performance of the proposed AirComp-based edge inference system over Cloud-RAN. 

\subsection{Experiment Settings}
\subsubsection{Network Settings}
We consider a Cloud-RAN network with $K = 20$ single-antenna devices and $M = 4$ RRHs.  The number of antennas will be stated later. The devices and RRHs are both randomly and independently located in a circular area with an inner radius of $100m$ and an outer radius of $500m$. The channel is modeled as the small-scale fading coefficients multiplied by the square root of the path loss, i.e., $\mathbf{h}_{k,m} = 10^{-pl(d)/20}\mathbf{s}_{k,m}$,  where $pl(d)$ is the path loss in $dB$ given as $30.6+36.7\log_{10} (d)$ and $d$ (in meter) is the distance between the device $k$ and RRH $m$. The small-scale fading coefficients $\{\mathbf{s}_{k,m}\}$ are assumed to follow the standard complex Gaussian distribution, i.e., $\mathbf{s}_{k,m} \sim \mathcal{CN}(0,\mathbf{I}),\; \forall (k,m)$. The power spectral density of the background noise at each RRH is set as $-169$ dBm/Hz and the noise figure is $7$ dB. All numerical results are averaged over $50$ trails.

\subsubsection{Inference Task}
 We perform two inference tasks, one on the human motion dataset \cite{li2021wireless} and the other on the Fashion MNIST dataset \cite{DBLP:journals/corr/abs-1708-07747}. The human motion dataset contains $6400$ training samples and $1600$ testing samples of $4$ different human motions, i.e.,  \textit{child walking, child pacing, adult pacing} and \textit{adult walking}.  The heights of children and adults are assumed to be uniformly distributed in the interval [$0.9$m, $1.2$m] and [$1.6$m, $1.9$m], respectively. The speed of standing, walking, and pacing are $0$ m/s, $0.5H$ m/s, and $0.25H$ m/s, respectively, where $H$ is the height value. The heading of the moving human is set to be uniformly distributed in $[-180^{\circ}$, $180^{\circ}]$.
In the human motion dataset, each edge device transmits the \emph{frequency-modulated continuous-wave} (FMCW) consisting of multiple up-ramp chirps for sensing. The reflected echo signals are sampled and arranged into a two-dimensional data matrix that contains the motion information of the interesting target polluted by the ground clutter and noise. The obtained data matrix is applied to a singular value decomposition (SVD) based linear filter for clutter elimination and is flattened into a $1520$-dimensional vector.  The Fashion MNIST dataset comprises $60,000$ training images and $10,000$ testing images with $10$ different fashion products such as T-shirts and  Trousers. 

\subsubsection{Inference Model}
 Two commonly used AI models, i.e., SVM and MLP neural networks are considered for the inference task. In the training process, the human motion dataset and Fashion MNIST dataset retain $12$ and $50$ principal components respectively, which also determines the input dimension of SVM and MLP.  It is sufficient since the proportion of variance contributed by the retained principal components accounts for more than $70$\% \cite{rea2016many}.  % It is reasonable since the proportion of human motion data samples with $1520$ dimensions retains $12$ principal components. There is usually no well-defined standard for determining the number of principal components. However, as a ``rule of thumb", a suitable cutoff point is to retain principal components that capture $70$\% or $90$\% of the variation \cite{rea2016many}. 
The one-vs-one strategy is employed in SVM, where a unique classifier is trained for each different pair of labels. This results in $6$ and $45$ binary classifiers for the human motion and Fashion MNIST datasets, respectively.  Each classifier uses hinge loss as the loss function and the sequential minimal optimization (SMO) algorithm as the optimization solver.  As for MLP, the neural network model consists of two hidden layers and is the same for both two datasets.  These two layers have $80$ and $40$ neurons and the ReLU is used as activation function. The network model adopts the LBFGS algorithm to minimize the cross-entropy loss.  The entire training process terminates after the  $16$-th iteration for the human motion dataset and the $1000$-th iteration for the Fashion MNIST dataset. These models are trained without any distortion, i.e., sensing clutter, quantization and noise distortion. The testing dataset is distorted by the clutter, quantization and noise introduced by the sensing and communication process. Although the training data used here are noise-free, it has been shown that the result of PCA on noisy data is similar to the result of PCA on noise-free data when the data and noise are independent \cite{DBLP:journals/tip/KhalilianB13}.

% \begin{enumerate}
%     \item \textit{Inference Dataset}: We perform the inference task on the human motion datasets proposed in \cite{li2021wireless} where sensing simulator is adopted to produce various high-fidelity human motions. In our simulation, we only consider to identify four distinct human motions, i.e., \textit{child walking, child pacing, adult pacing} and \textit{adult walking}. The entire dataset contains $8000$ samples, $2000$ for each class. Following the setting in \cite{https://doi.org/10.48550/arxiv.2206.05949}, the heights of children and adults are assumed to be uniformly distributed in intervals [$0.9$m, $1.2$m] and [$1.6$m, $1.9$m], respectively. The speed of standing, walking, and pacing are $0$ m/s, $0.5H$ m/s, and $0.25H$ m/s, respectively, where $H$ is the height value. In addition, the heading of the moving human is set to be uniformly distributed in $[-180^{\circ}$, $180^{\circ}]$.
%     \item \textit{Inference Model}: SVM and MLP neural networks are two commonly used classification models in machine learning and they are considered for our inference task. Specifically, the neural network model consists of two hidden layers, each with $80$ and $40$ neurons. The human motion datasets are divided into non-overlapping training and test sets with $6400$ and $1600$ samples, respectively. Both models are trained using the same training data regardless of channel and data distortion.
% \end{enumerate}

\subsection{Convergence of the Proposed Algorithm}
\begin{figure}
    \centering
    \includegraphics[width=0.35\textwidth]{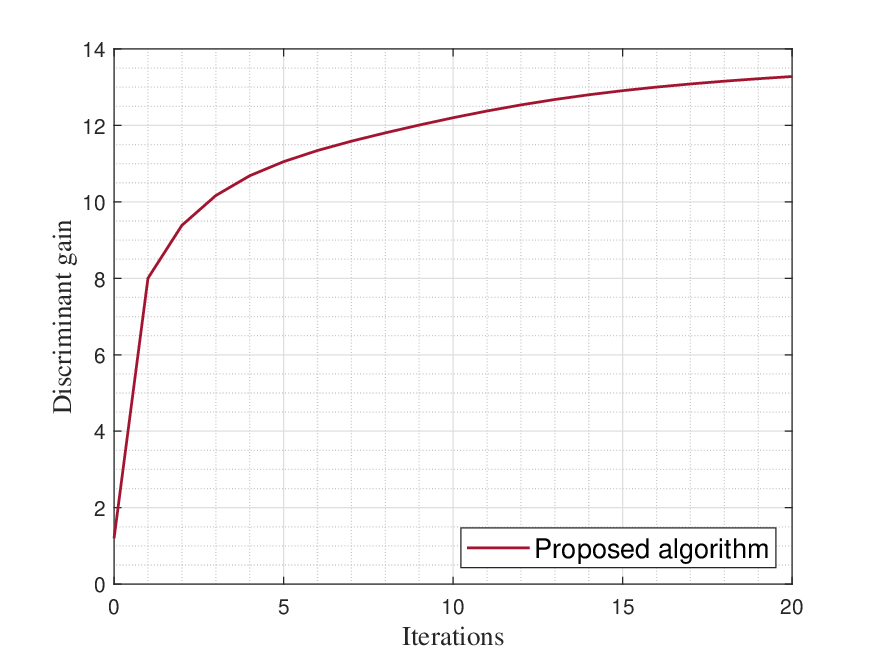}
    \captionsetup{font=small}
    \caption{Convergence behavior of Proposed Algorithm.}
    \label{fig:conv_behavior}
\end{figure}
\begin{figure}
    \centering
    \includegraphics[width=0.35\textwidth]{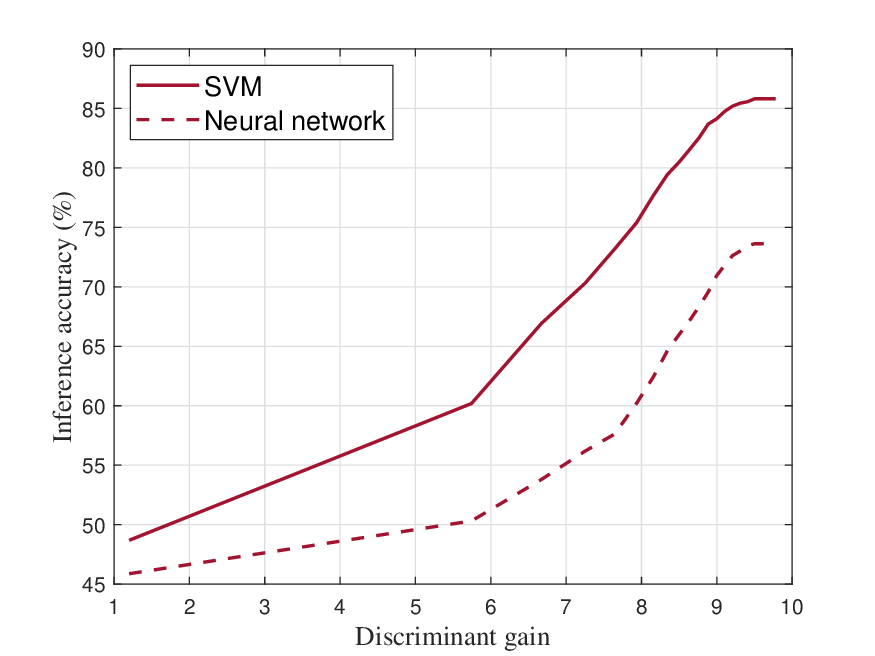}
    \captionsetup{font=small}
    \caption{Inference accuracy versus discriminant gain.}
    \label{fig:dg v.s. acc}
\end{figure}

In this part, we show the convergence behavior of the proposed algorithm and outline the relationships between the discriminant gain and inference accuracy. In Fig. \ref{fig:conv_behavior}, we plot the discriminant gain achieved by the proposed algorithm with power constraint $P = 23$ dBm. It is observed that the  discriminant gain increases quickly and converges within a few iterations. This demonstrates the efficiency of the proposed algorithm for joint optimization. Besides, the relation between discriminant gain and the instantaneous inference accuracy is illustrated in Fig. \ref{fig:dg v.s. acc}. It can be found from the figure that the inference accuracy is monotonically increasing with an increasing value of discriminant gain for both models, which verifies the effectiveness of the latter. 

\begin{figure*}
     \centering
     \begin{subfigure}{0.495\textwidth}
        \centering
         \includegraphics[scale=0.5,trim={5pt 0pt 20pt 20pt},clip]{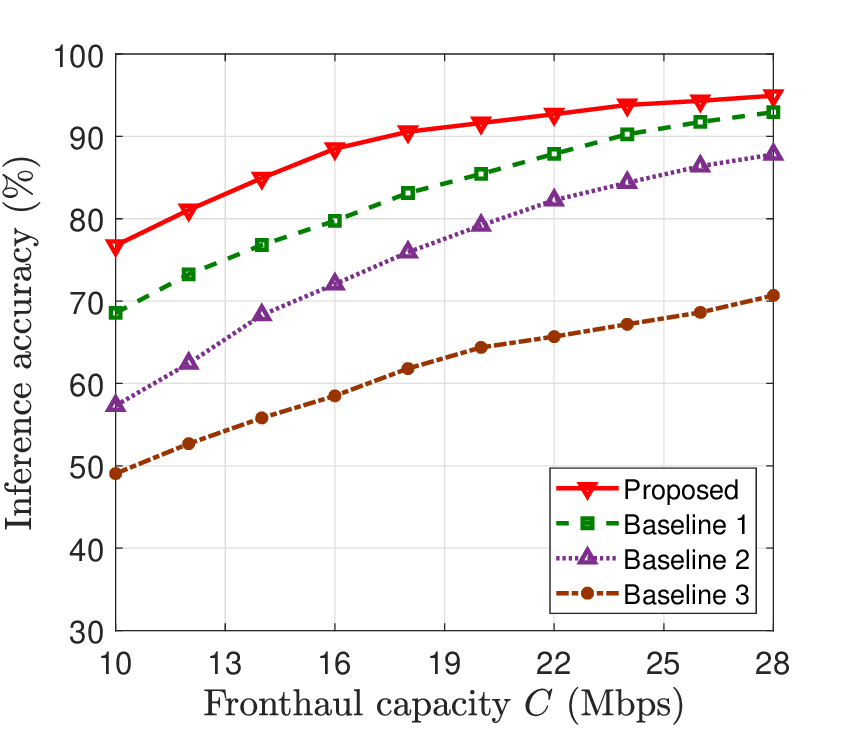}
         \captionsetup{font=small}
         \caption{Inference accuracy of SVM versus fronthaul capacity.}
         \label{fig:sensing_capacity_a}
        \end{subfigure}
        \hspace{-4.3em}
         \begin{subfigure}{0.495\textwidth}
        \centering
         \includegraphics[scale=0.5,trim={0pt 0pt 20pt 20pt},clip]{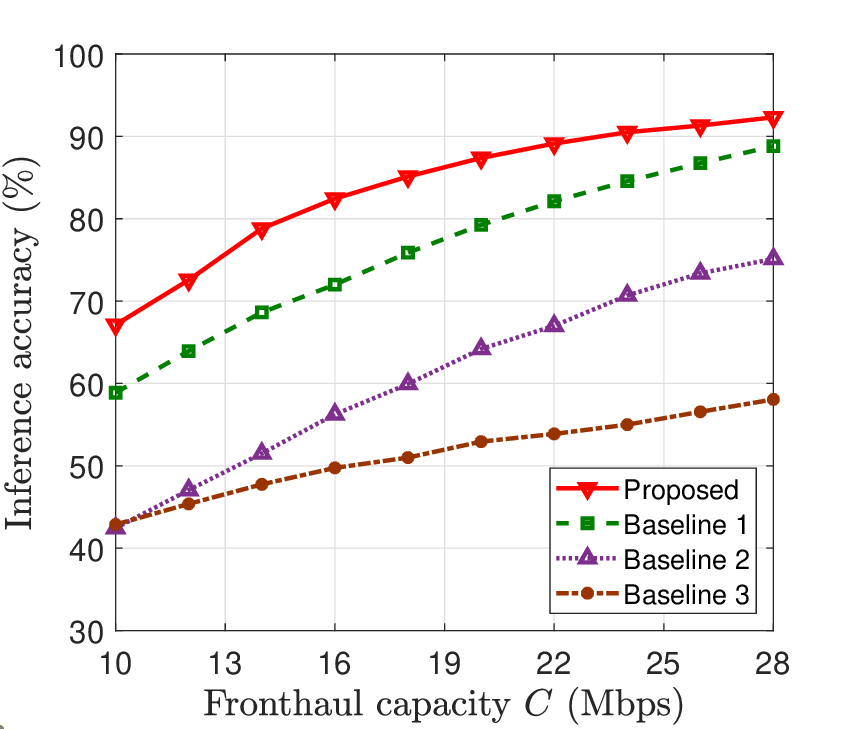}
         \captionsetup{font=small}
         \caption{Inference accuracy of MLP versus fronthaul capacity.}
         \label{fig:sensing_capacity_b}
        \end{subfigure}
        \captionsetup{font=small}
        \caption{Inference accuracy comparison among different models under different fronthaul capacities for the human motion dataset with $N=4$.}
        \label{fig:sensing_capacity}
\end{figure*}

\begin{figure*}
     \centering
     \begin{subfigure}{0.495\textwidth}
        \centering
         \includegraphics[scale=0.5,trim={0pt 0pt 20pt 20pt},clip]{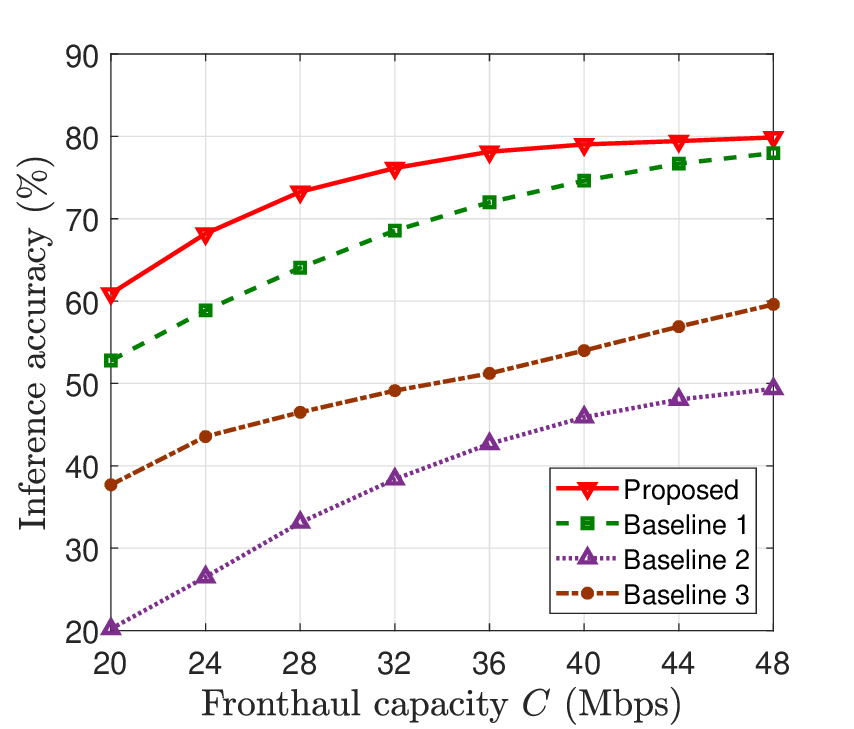}
         \captionsetup{font=small}
         \caption{Inference accuracy of SVM versus fronthaul capacity.}
         \label{fig:fmnist_capacity_a}
        \end{subfigure}
        \hspace{-4.3em}
         \begin{subfigure}{0.495\textwidth}
        \centering
         \includegraphics[scale=0.5,trim={5pt 0pt 20pt 20pt},clip]{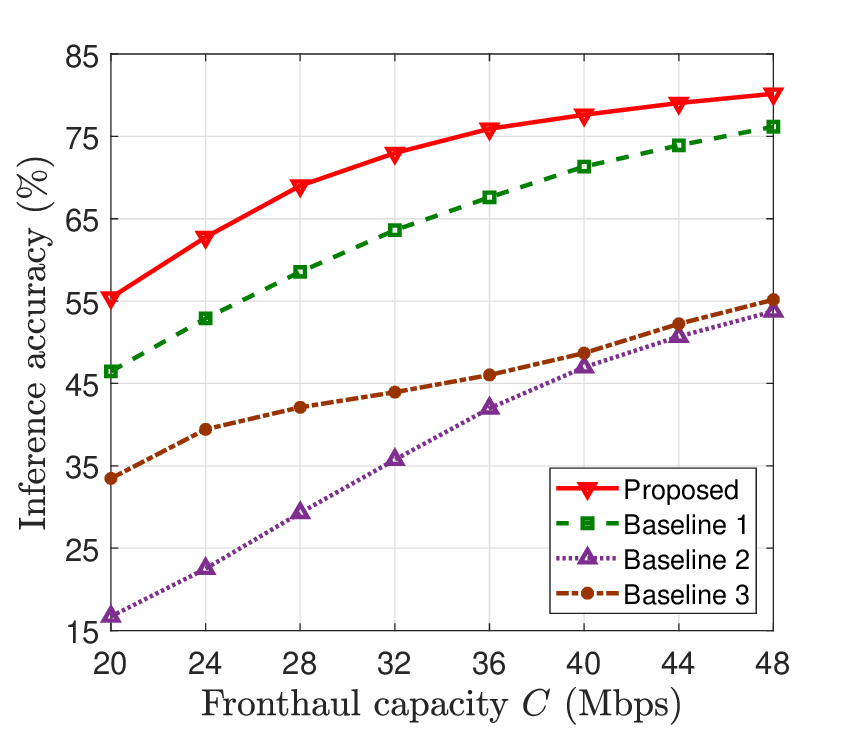}
         \captionsetup{font=small}
         \caption{Inference accuracy of MLP versus fronthaul capacity.}
         \label{fig:fmnist_capacity_b}
        \end{subfigure}
        \captionsetup{font=small}
        \caption{Inference accuracy comparison among different models under different fronthaul capacities for the Fashion MNIST dataset with $N=2$.}
        \label{fig:fmnist_capacity}
\end{figure*}

\begin{figure*}
\centering
\begin{subfigure}{0.495\textwidth}
    \centering
    \includegraphics[scale=0.5,trim={0pt 0pt 20pt 20pt},clip]{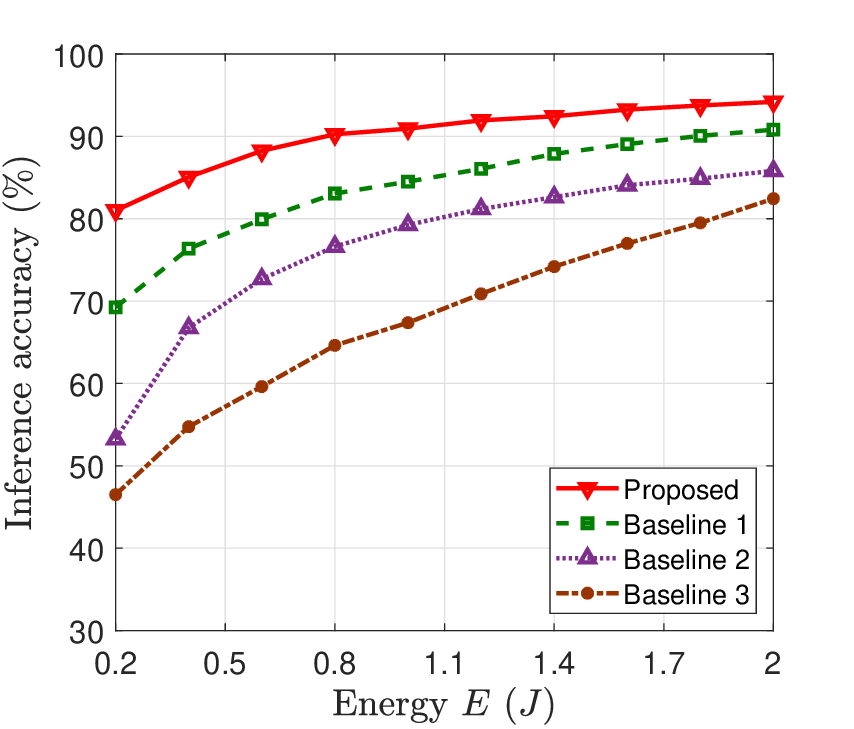}\vspace{0cm}
    \captionsetup{font=small}
    \caption{Inference accuracy of SVM versus energy.}\label{fig:sensing_energy_a}
\end{subfigure}
\hspace{-4.3em}
\begin{subfigure}{0.495\textwidth}
    \centering
    \includegraphics[scale=0.5,trim={0pt 0pt 20pt 20pt},clip]{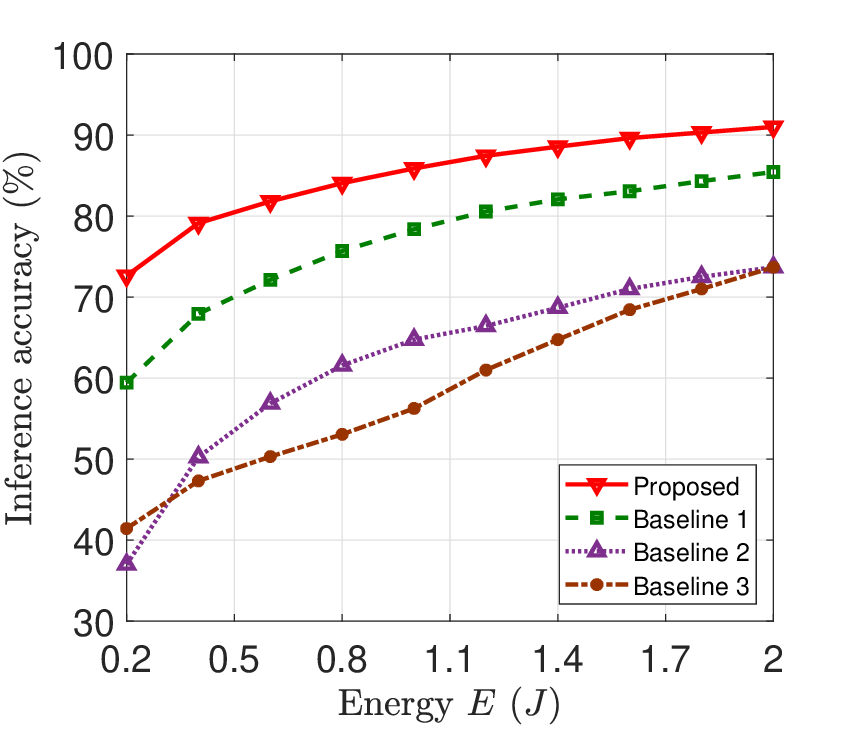}\vspace{0cm}
    \captionsetup{font=small}
    \caption{Inference accuracy of MLP versus energy.}
    \label{fig:sensing_energy_b}
    \end{subfigure}
    \captionsetup{font=small}
    \caption{Inference accuracy comparison among different models under different energy constraints for the human motion dataset with $N=4$.} \label{fig:sensing_energy}
\end{figure*}

\begin{figure*}
\centering
\begin{subfigure}{0.495\textwidth}
    \centering
    \includegraphics[scale=0.5,trim={5pt 0pt 20pt 20pt},clip]{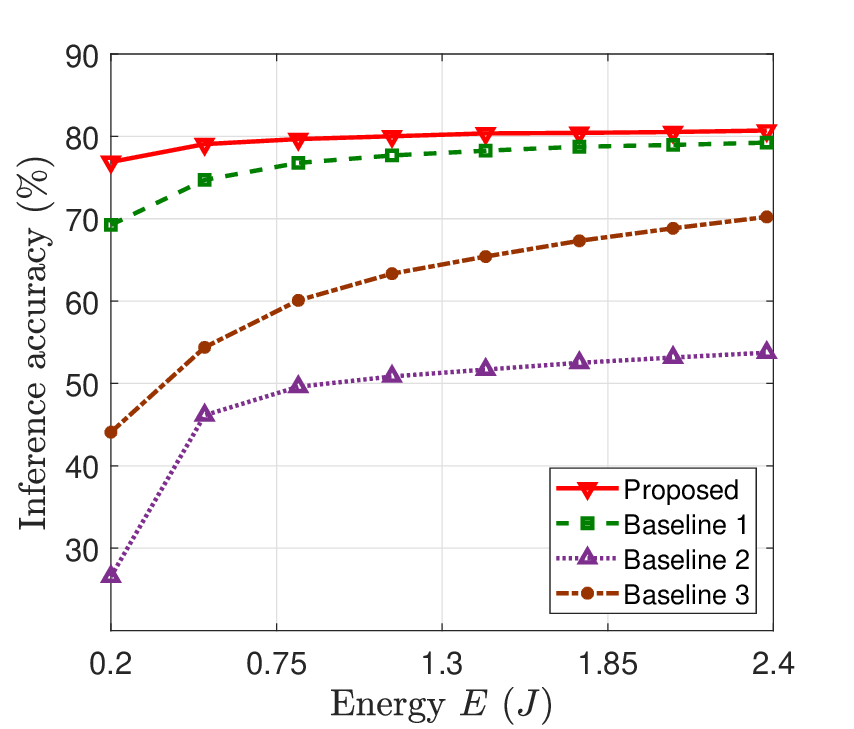}\vspace{0cm}
    \captionsetup{font=small}
    \caption{Inference accuracy of SVM versus energy.}\label{fig:fmnist_energy_a}
\end{subfigure}
\hspace{-4.3em}
\begin{subfigure}{0.495\textwidth}
    \centering
    \includegraphics[scale=0.5,trim={5pt 0pt 20pt 20pt},clip]{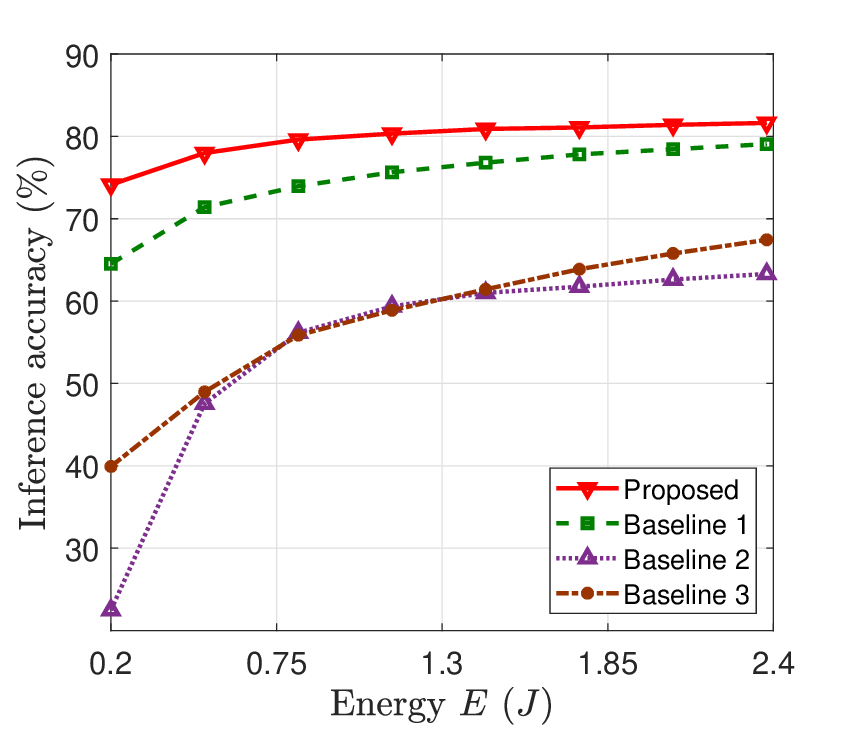}\vspace{0cm}
    \captionsetup{font=small}
    \caption{Inference accuracy of MLP versus energy.}
    \label{fig:fmnist_energy_b}
    \end{subfigure}
    \captionsetup{font=small}
    \caption{Inference accuracy comparison among different models under different energy constraints for the Fashion MNIST dataset with $N=2$.} \label{fig:fmnist_energy}
\end{figure*}
\subsection{Impact of Key System Parameters}
    In this part, we show the performance gain of joint optimization over other baseline methods under the wireless and fronthaul resource constraints and investigate the impact of various key system parameters. For ease of presentation, we refer to our proposed algorithm for jointly optimizing transmit precoding, quantization noise matrix and receive beamforming as \textbf{Proposed} and set the following schemes as baselines for comparison:
    \begin{itemize}
        \item \textbf{Baseline 1: Uniform quantization with joint optimization of transmit precoding and receive beamforming.}  In Baseline 1, the optimized portion of transmit precoding and receive beamforming follows Algorithm \ref{alg:1}. The CP performs uniform quantization at all antennas across all RRHs, i.e., setting $\mathbf{Q} = \lambda \mathbf{I}$, where scalar $\lambda$ can be easily be selected by binary search to exactly satisfy the capacity constraint \eqref{cons:1a}.  The transmit precoding and receive beamforming are jointly optimized.
        \item \textbf{Baseline 2: Uniform receive beamforming with joint optimization of transmit precoding and quantization matrix.} In Baseline 2, the optimized portion of the transmit precoding and quantization matrix follows Algorithm \ref{alg:1}. The receive beamforming is uniformly designed, i.e., setting $\mathbf{m}_d = \mathbf{1}$.
         \item {\color{blue}\textbf{Baseline 3:  Fixed transmit precoding with joint optimization of quantization matrix and receive beamforming.} In Baseline 3, we fixed the transmitting precoding $\{b_k(d)\}$ as the same value for all devices in all time slots. The set value does not violate energy and power constraints. Then the quantization and receive beamforming are jointly designed. }
    \end{itemize}
    
In the sequel, the proposed joint optimization scheme is compared with the above three baseline schemes.

\subsubsection{Inference Accuracy v.s. Fronthaul Capacity}
The inference accuracy of both models achieved by different schemes under various fronthaul capacity $C$ is shown in Fig. \ref{fig:sensing_capacity} and Fig. \ref{fig:fmnist_capacity}. It is observed that when fronthaul capacity increases in all cases, inference accuracy in all schemes improves.  Our proposed joint optimization achieves the best performance over Baselines $1$, $2$ and $3$. {\color{blue}  Particularly, the proposed scheme outperforms Baseline 3. Generally, the fixed transmit precoding design in Baseline 3 cannot capture the diverse importance levels of different feature elements on inference accuracy.}  Furthermore, it is also noticed that Baseline $1$ with uniform quantization consistently outperforms Baseline $2$. This indicates that optimization of beamforming on the CP can achieve more performance gain than that of quantization.
        
\subsubsection{Inference Accuracy v.s. Energy}
Fig. \ref{fig:sensing_energy} and Fig. \ref{fig:fmnist_energy} show the inference accuracy of both models achieved by different schemes under different energy thresholds. From the figure, the inference accuracy increases as the energy requirement is gradually relaxed. This is due to the fact that more energy suppresses the channel noise and thus the discriminant gain is enhanced. In addition, similar to the case of the fronthaul capacity, we can also conclude that Baseline $1$ outperforms Baseline $2$.

The extensive experimental results presented above show the priority of the proposed joint optimization scheme and verify our theoretical analysis.

\section{Conclusion}
In this paper, we implemented task-oriented communication for multi-device cooperative edge inference over a Cloud-RAN based wireless network, where the edge devices upload extracted features to the CP using AirComp.  The design of AirComp does not follow the previous criterion of MMSE, but directly adopts the inference accuracy as the design goal. Particularly, since the instantaneous inference accuracy is intractable, an approximate metric called discriminant gain is adopted as the alternative. This task-oriented communication systems are ultimately modeled as an optimization problem that maximizes discriminant gain. To address this problem, we develop an efficient iterative algorithm to solve this non-convex problem by applying variable transformation, SCA and alternating optimization techniques. Extensive numerical results show that our proposed optimization algorithm can achieve higher inference performance and the effectiveness of the proposed Cloud-RAN network architecture for cooperative inference was also verified.

This work opens several research directions. One is the device scheduling at each CP for selecting only a subset of devices. The other is to overcome the shortages such as pilot overheads and channel estimation errors caused by the channel estimations of the large number of wireless links.

\section{Appendix}
\subsection{Proof of Lemma \ref{lma:1}} 
As mentioned in \eqref{pdf}, the ground-true feature vector can be written as the average of $L$ independent Gaussian random variables,
\begin{equation}
\begin{aligned}
    \mathbf{\tilde{x}} = \frac{1}{L} \sum_{\ell = 1}^L  \tilde{\mathbf{x}}_{\ell},
\end{aligned}
\end{equation}
where $\tilde{\mathbf{x}}_{\ell} \sim \mathcal{N}(\bm{\mu}_\ell, \bm{\Sigma})$.

Then taking back into \eqref{feature_vector},  the ground-true feature vector $\mathbf{\Tilde{x}}$ becomes
\begin{equation}
\begin{aligned}
    \mathbf{\Tilde{x}}_k =  \frac{1}{L} \sum_{\ell = 1}^L \tilde{\mathbf{x}}_{\ell} +  \mathbf{\Tilde{e}}_k = \frac{1}{L} \sum_{\ell = 1}^L \tilde{\mathbf{x}}_{\ell, k},
\end{aligned}
\end{equation}
where $\tilde{\mathbf{x}}_{\ell, k} = \tilde{\mathbf{x}}_{\ell} + \mathbf{\Tilde{e}}_k$. Thus we can obtain the distribution of  $\tilde{\mathbf{x}}_{\ell, k}$ as 
\begin{equation}
\begin{aligned}
    \tilde{\mathbf{x}}_{\ell, k} \sim \mathcal{CN}(\bm{\mu}_\ell, \bm{\Sigma}+ \varepsilon_k^2 \mathbf{I} ), 1 \leq \ell \leq L.
\end{aligned} 
\end{equation}
Finally, the distribution of local feature vector $\tilde{\mathbf{x}}_{\ell, k}$ of device $k$ is given by
\begin{equation}
\begin{aligned}
     f( \mathbf{\Tilde{x}}_k ) = \dfrac{1}{L}\sum\limits_{\ell=1}^L \mathcal{N}(\bm{\mu}_{\ell},\bm{\Sigma}+\varepsilon_k^2 \mathbf{I} ), \ \forall k\in \mathcal{K}. 
\end{aligned}
\end{equation}

\subsection{Proof of Lemma \ref{lma:2}}
Following the same approach as lemma \ref{lma:1} but taking the element-wise version, the estimated element can be written as

\begin{equation}
\begin{aligned}
    \hat{s}(d) &= \frac{1}{L}   \sum\limits_{k=1}^K \sum_{\ell = 1}^L c_k(d) \mathbf{\tilde{x}}_{\ell}(d) + \sum\limits_{k=1}^K c_k(d) \mathbf{\Tilde{e}}_k(d) + n(d) \\
    &= \frac{1}{L} \sum_{\ell = 1}^L \mathbf{\tilde{x}}_{\ell, s}(d),
\end{aligned}
\end{equation}
where $\mathbf{\tilde{x}}_{\ell, s}(d) = \sum\limits_{k=1}^K c_k(d) \mathbf{\tilde{x}}_{\ell}(d) +  \sum\limits_{k=1}^K c_k(d) \mathbf{\Tilde{e}}_k(d) +  n(d)$. 

Thus, we can obtain the distribution of  $\mathbf{\tilde{x}}_{\ell, s}(d)$ as
\begin{equation}
\begin{aligned}
\mathbf{\tilde{x}}_{\ell, s}(d) &\sim \mathcal{N}\left( \sum\limits_{k=1}^K c_k(d) \bm{\mu}_{\ell}(d), \right.  \\
&\left. \Big(\sum\limits_{k=1}^K c_k(d) \Big)^2 \sigma^2_d  + \sum\limits_{k=1}^K c_k^2(d) \varepsilon_k^2 +  \sigma^2 \right),  1 \leq \ell \leq L.
\end{aligned}
\end{equation}

Finally, the distribution of  the aggregation signal $\hat{s}(d)$ is given by
\begin{equation}
\begin{aligned}
    \hat{s}(d) \sim \frac{1}{L}\sum\limits_{\ell=1}^L \mathcal{N}\left( \hat{\bm{\mu}}_{\ell}(d), \hat{\sigma}^2_d \right), \ \forall d \in \mathcal{D}.
\end{aligned}
\end{equation}

\subsection{Proof of Lemma \ref{lma:3}}
    Suppose that the new problem $\mathscr{P}'_1$ has an optimal solution: $\{ \mathbf{A}^*, \mathbf{C}^*, \mathbf{M}^*,\mathbf{Q}^* \}$, these exists a  $d' \in [1, D]$ such that the inequality \eqref{cons:333b} strictly holds, i.e., 
\begin{equation} \label{cons:3333b}
    \Lambda(\{c_k^*(d')\}, \{\mathbf{m}_{d'}^*\}, \mathbf{Q}^*) < \Gamma_1(\alpha^*(d'), \{c_k^*(d')\}), 
\end{equation}
% \begin{equation} \label{cons:3333b}
% \small
%    \alpha^*(d') < \frac{\frac{2}{L(L-1)} \Big( \sum\limits_{k=1}^K c_k^*(d') \Big)^2 \sum\limits_{\ell=1}^L \sum\limits_{\ell < \ell'}\left({\bm{\mu}}_{\ell}(d')- {\bm{\mu}}_{\ell'}(d') \right)^2}{\Big( \sum\limits_{k=1}^K c_k^*(d') \Big)^2 \sigma^2_{d'} + \sum\limits_{k=1}^K c_k^{*2}(d') \varepsilon_k^2 + \frac{1}{2} (\mathbf{m}^*_{d'})^{\sf H} \left(\sigma_z^2 \mathbf{I} + \mathbf{Q^*} \right) \mathbf{m}^*_{d'}}.
% \end{equation}  

% Then, from the continuity of linear function on the left-hand side of, under a fixed $\{\mathbf{C}^*, \mathbf{M}^*,\mathbf{Q}^*\}$, there always exists a number $\eta > 0$ such that
% \begin{equation}
%     \alpha^*_+(d') = (1 + \eta)  \; \alpha^*(d') > \alpha^*(d'),
% \end{equation}
% which leads to 

% \begin{equation}
% \small
% \begin{aligned}
% \alpha^*(d') &< \alpha^*_+(d') < \\
% & \frac{\frac{2}{L(L-1)} \Big( \sum\limits_{k=1}^K c_k^*(d') \Big)^2 \sum\limits_{\ell=1}^L \sum\limits_{\ell < \ell'}\left({\bm{\mu}}_{\ell}(d')- {\bm{\mu}}_{\ell'}(d') \right)^2}{\Big( \sum\limits_{k=1}^K c_k^*(d') \Big)^2 \sigma^2_{d'} + \sum\limits_{k=1}^K c_k^{*2}(d') \varepsilon_k^2 + \frac{1}{2} (\mathbf{m}^*_{d'})^{\sf H} \left(\sigma_z^2 \mathbf{I} + \mathbf{Q^*} \right) \mathbf{m}^*_{d'}}.
% \end{aligned}
% \end{equation}

Based on the continuity of inversely proportional function about $\alpha(d)$ on the right-hand side of \eqref{cons:3333b},  under a fixed $\{\mathbf{C}^*, \mathbf{M}^*,\mathbf{Q}^*\}$, there always exists a number $\eta > 0$ such that 
\begin{equation}
    \alpha^*_{+}(d') = (1 + \eta)  \; \alpha^*(d') > \alpha^*(d'),
\end{equation}
which leads to 
\begin{equation} 
\begin{aligned}
    \Lambda(\{c_k^*(d')\}, \{\mathbf{m}_{d'}^*\}, \mathbf{Q}^*) &< \Gamma_1(\alpha^*_{+}(d'), \{c_k^*(d')\}) \\ 
    &< \Gamma_1(\alpha^*(d'), \{c_k^*(d')\}).
\end{aligned}
\end{equation}

By substituting $\alpha^*_+(d')$ into $\mathscr{P}_1'$, the value of the objective function can be increased further. This is a contradiction of the fact that  $\alpha^*(d')$ is the optimal solution of problem $\mathscr{P}_1'$. Thus, the problem extended the constraint \eqref{cons:33b} achieves the same optimal solution as $\mathscr{P}_1$.

\subsection{Proof of Lemma \ref{lma:4}}
    Given a set of variables $\{\mathbf{C}, \mathbf{M}\}$ satisfying constraint \eqref{cons:1b}, it is always possible to let $\beta_{k,d} = \frac{c_k^2(d)}{\left| \mathbf{m}_d^{\sf H} \mathbf{h}_k \right|^2}, \forall k \in \mathcal{K}, \forall d \in \mathcal{D} $, then constraints \eqref{cons:2a} \eqref{cons:2b} holds.
    Given a set of variables $\{\mathbf{B}, \mathbf{C}, \mathbf{M}\}$ satisfying constraints \eqref{cons:2a} \eqref{cons:2b}, then immediately constraint \eqref{cons:1223}  holds by simple algebra operations. Simultaneously summing $K$ and $D$ on both sides of the inequality in constraint \eqref{cons:1223} and combining with \eqref{cons:2b}, the inequality \eqref{cons:1b} is derived.

\bibliography{reference}
\bibliographystyle{ieeetr}

\end{document}